\documentclass{ws-jbcb}

\usepackage{graphicx}
\newcommand{\myparagraph}[1]{{\smallskip\noindent{\bf #1}}}

\newcommand{\hatZ}{{\hat Z}}
\newcommand{\sqrts}{{\sqrt s}}
\newcommand{\tildeX}{{\tilde X}}
\newcommand{\tildeZ}{{\tilde Z}}

\newcommand{\barz}{{\bar z}}

\newcommand{\barZ}{{\bar Z}}

\newcommand{\calA}{{\mathcal A}}
\newcommand{\calB}{{\mathcal B}}
\newcommand{\calC}{{\mathcal C}}
\newcommand{\calD}{{\mathcal D}}

\newcommand{\calP}{{\mathcal P}}

\newcommand{\starx}{x^\ast}
\newcommand{\starz}{z^\ast}

\newcommand{\sqrtstarz}{\sqrt{z^\ast}}

\newcommand{\braced}[1]{{ \left\{ #1 \right\} }}
\newcommand{\angled}[1]{{ \left\langle #1 \right\rangle }}

\newcommand{\ceiling}[1]{{ \lceil #1 \rceil }}
\newcommand{\floor}[1]{{ \lfloor #1 \rfloor }}

\newcommand{\half}{{\mbox{$\frac{1}{2}$}}}
\newcommand{\onefourth}{{\mbox{$\frac{1}{4}$}}}
\newcommand{\oneeighth}{{\mbox{$\frac{1}{8}$}}}
\newcommand{\avgn}{{\mbox{$1\over n$}}}

\newcommand{\minIP}{\mbox{\sf MinIP}}
\newcommand{\minLP}{\mbox{\sf MinLP}}
\newcommand{\maxIP}{\mbox{\sf MaxIP}}
\newcommand{\maxLP}{\mbox{\sf MaxLP}}

\newcommand{\avgIP}{\mbox{\sf AvgIP}}
\newcommand{\avgLP}{\mbox{\sf AvgLP}}
\newcommand{\XthreeC}{\mbox{\sf X3C}}
\newcommand{\PP}{\mbox{$\mathbb{P}$}}
\newcommand{\NP}{\mbox{$\mathbb{NP}$}}
\newcommand{\RP}{\mbox{$\mathbb{RP}$}}
\newcommand{\DTIME}{\mbox{$\mathbb{DTIME}$}}

\newcommand{\RCM}{\mbox{\rm RCM}}
\newcommand{\RCMtwo}{\mbox{\rm RCM2}}
\newcommand{\RDM}{\mbox{\rm RDM}}
\newcommand{\RCA}{\mbox{\rm RCA}}
\newcommand{\RCAtwo}{\mbox{\rm RCA2}}

\newcommand{\malymin}{{\mbox{\tiny\it min}}}
\newcommand{\malymax}{{\mbox{\tiny\it max}}}
\newcommand{\malyavg}{{\mbox{\tiny\it avg}}}
\newcommand{\malyRCM}{{\mbox{\tiny\rm RCM}}}
\newcommand{\malyRCMtwo}{{\mbox{\tiny\rm RCM2}}}
\newcommand{\malyRDM}{{\mbox{\tiny\rm RDM}}}

\newcommand{\malyRCAtwo}{{\mbox{\tiny\rm RCA2}}}

\newcommand{\Cmin}[1]{\mbox{$\calC^{#1}_\malymin$}}

\newcommand{\Dmax}[1]{\mbox{$\calD^{#1}_\malymax$}}

\newcommand{\Cavg}[1]{\mbox{$\calC^{#1}_\malyavg$}}
\newcommand{\Davg}[1]{\mbox{$\calD^{#1}_\malyavg$}}
\newcommand{\BCPCmin}{\mbox{\mbox{\sf BCP\_}\Cmin{}}}

\newcommand{\BCPDmax}{\mbox{\mbox{\sf BCP\_}\Dmax{}}}

\newcommand{\BCPCavg}{\mbox{\mbox{\sf BCP\_}\Cavg{}}}
\newcommand{\BCPDavg}{\mbox{\mbox{\sf BCP\_}\Davg{}}}

\newcommand{\Prob}{\mbox{\rm Pr}}
\newcommand{\Exp}{\mbox{\rm Exp}}
\newcommand{\dnaseq}[1]{{\mbox{\rm\small\tt #1}}}
\begin{document}

\markboth{Qi Fu, Elizabeth Bent, James Borneman, Marek Chrobak, and Neal E. Young}
{Algorithmic Approaches to Selecting Control Clones
        in DNA Array Hybridization Experiments}

\catchline{}{}{}{}{}
%
%%%%%%%%%%%%%%%%%%%%%%%%%%%%%%%%%%%%%%%%%%%%%%%%%%%%%%%%%%%%%%%%%%%%

\title{
ALGORITHMIC APPROACHES TO SELECTING CONTROL CLONES
	IN DNA ARRAY HYBRIDIZATION EXPERIMENTS
\thanks{A preliminary version of this paper appeared in the
Proceedings of Asia-Pacific Bioinformatics Conference 2007.
Journal version published as 
Journal of Bioinformatics and Computational Biology 5(4) 937--961, 2007,
DOI 10.1142/S0219720007002977, World Scientific Publishing Company}}

\author{%
QI FU%
\footnotemark[1]
,
ELIZABETH BENT%
\footnotemark[2]
,
JAMES BORNEMAN%
\footnotemark[2]
,\\
MAREK CHROBAK%
\footnotemark[1]
\and
NEAL E. YOUNG%
\footnotemark[1]
}

\address{Department of Computer Science \footnotemark[1] ,
	    Department of Plant Pathology \footnotemark[2] ,\\
            University of California,\\
            Riverside, CA 92521. U.S.A.\\
	    (qfu, marek, neal)@cs.ucr.edu \footnotemark[1]\\
	    (bente, james.borneman)@ucr.edu \footnotemark[2]
	    }

\maketitle

%\begin{history}
%\received{(Day Month Year)}
%\revised{(Day Month Year)}
%\accepted{(Day Month Year)}
%\comby{(xxxxxxxxxx)}
%\end{history}

%%%%%%%%%%%%%%%%%%%%%%%%%%%%%%%%%%%%%%%%%%%%%%%%%%%%%%%%
%%%%%%%%%%%%%%%%%%%%%%%%%%%%%%%%%%%%%%%%%%%%%%%%%%%%%%%%
%%%%%%%%%%%%%%%%%%%%%%%%%%%%%%%%%%%%%%%%%%%%%%%%%%%%%%%%

\begin{abstract}
We study the problem of selecting control clones
in DNA array hybridization experiments. The
problem arises in the OFRG
method for analyzing microbial communities.
The OFRG method performs classification of rRNA
gene clones using binary fingerprints created from
a series of hybridization experiments, where
each experiment consists of
hybridizing a collection of arrayed clones with
a single oligonucleotide probe.
This experiment produces analog signals, one
for each clone, which then need to be classified,
that is, converted into binary values $1$ and $0$
that represent hybridization and non-hybridization events.
In addition to the sample rRNA gene clones, the array
contains a number of control clones needed to calibrate the
classification procedure of the hybridization signals. These
control clones must be selected with care to
optimize the classification process.
We formulate this as a
combinatorial optimization problem called
\emph{Balanced Covering}. We prove that the
problem is {\NP}-hard, and we show some results
on hardness of approximation.
We propose approximation algorithms
based on randomized rounding and
we show that, with high probability,
our algorithms approximate well the optimum solution.
The experimental results confirm that the
algorithms find high quality control clones.
The algorithms have been implemented and are
publicly available as part of the software package
called CloneTools.
\end{abstract}

%%%%%%%%%%%%%%%%%%%%%%%%%%%%%%%%%%%%%%%%%%%%%%%%%%%%%%%%
%%%%%%%%%%%%%%%%%%%%%%%%%%%%%%%%%%%%%%%%%%%%%%%%%%%%%%%%
%%%%%%%%%%%%%%%%%%%%%%%%%%%%%%%%%%%%%%%%%%%%%%%%%%%%%%%%

\newpage

\section{Introduction}

\myparagraph{Background.}
We study the problem of selecting control clones
for DNA array hybridization experiments.
The specific version of the problem that we
address arises in the context of the OFRG
(Oligonucleotide Fingerprinting of Ribosomal RNA Genes)
method,
that we describe below, although our approach is
also relevant to other applications of DNA
microarray technology.

OFRG (\cite{Andres04}, \cite{Jampa05}, \cite{Valin04}, 
\cite{Valin02a}, \cite{Valin02b})
is a technique for analyzing microbial communities
that classifies rRNA gene clones into taxonomic
clusters based on binary fingerprints
created from hybridizations with
a collection of oligonucleotide probes.
More specifically, in OFRG, clone libraries
from a sample under study (e.g., fungi or bacteria
from an environmental sample) are constructed
using PCR primers. These cloned
rRNA gene fragments are immobilized on nylon membranes and
then subjected to a series of hybridization experiments,
with each experiment using a single radiolabeled
DNA oligonucleotide probe.
This experiment produces analog signals, one
for each clone, which then need to be classified,
that is, converted into binary values $1$ and $0$
that represent hybridization and non-hybridization events.
Overall,
this process creates a hybridization fingerprint
for each clone, which is a vector of binary values
indicating which probes bind with this clone and
which do not. The clones are then
identified by clustering their hybridization fingerprints
with those of known sequences and by nucleotide sequence
analysis of representative clones within a cluster.

In addition to sample clones, the array contains a
number of \emph{control clones}, with known nucleotide sequences,
used to calibrate the classification procedure
of hybridization signals. Consider a hybridization
experiment with a probe $p$. Signal intensities from
its hybridizations with the control clones produce two
distributions: one from control clones that match $p$
(e.g., they contain $p$ or $p$'s reverse complement and thus
should hybridize with it)
and the other from control clones that do not.
This information is used to determine, via 
appropriate statistical techniques,
$p$'s signal intensity threshold $t$. Once $t$
has been determined, we can classify signal
intensities for sample clones as follows: 
signals above $t$ are
interpreted as $1$'s (hybridization events)
while those below are represented by $0$'s (non-hybridizations).

The quality of information obtained
from hybridizations depends critically on the
accuracy of the signal classification process. In particular,
the control clones should be more or less equally
distributed in terms of their ability to bind or not
bind with each probe from a given probe set.
In prior OFRG work, control clones were selected arbitrarily,
often producing control clones with
very skewed distribution of binding/non-binding with
some probes. As an example, from a set of 100 control clones,
only two might bind with a specific probe. The signal
classification for this probe would be very unreliable,
as it would be based on signal intensities from
hybridization with only two control clones.

%%%%%%%%%%%

\myparagraph{Problem formulation.}
Our control-clone selection problem can be then formulated as follows:
We are given a collection $C$ of candidate control clones
and a set $P$ of oligonucleotide probes to be used in
the hybridization experiments.
From among the candidate clones in $C$,
we want to select a set $D\subseteq C$ of $s$ control clones
such that each probe in $P$ hybridizes
with roughly half of the clones in $D$.

This gives rise to a combinatorial optimization
problem that we call
\emph{Balanced Covering}\footnote{
There have been some discussions on the
\emph{Balanced Set Cover} (see \cite{berger1989ena}, \cite{gargano2003mgo})
problem, however, they are not directly related to
\emph{Balanced Covering} problem discussed
in this paper.}.
The instance is given as a pair $\angled{G,s}$, where
$G = (C,P,E)$ is a bipartite graph and
$s\le |C|$ is an integer.
$C$ represents the clone set, $P$ is the probe set,
and the edges in $E$ represent potential
hybridizations between clones and probes,
that is, $(c,p)\in E$ iff $c$ contains $p$ or the
reverse complement of $p$.
For $p\in P$ and $D\subseteq C$,
let $\deg_D(p)$ be the number of
neighbors of $p$ in $D$ (that is, the number of
clones in $D$ that hybridize with $p$).
Throughout the paper, unless stated otherwise,
by $m$ we will denote the cardinality of $C$ and by $n$ the
cardinality of $P$.

%%%%%%%%%%%

\myparagraph{Example.} 
We illustrate the concept with a small example.
(The realistic data sets are typically considerably larger.) Let
$P = \braced{p_1,p_2,...,p_7}$ be the following probe set:

\begin{center}
\begin{tabular}{c c c c c c c}
$p_1$ & $p_2$ & $p_3$ & $p_4$ & $p_5$ & $p_6$ & $p_7$ \\
\dnaseq{CTGGC} & \dnaseq{TACAT} & \dnaseq{CGGCG} & \dnaseq{GCTGG}
    & \dnaseq{CGCTA} & \dnaseq{GCCTA} & \dnaseq{ATACA}
\end{tabular}
\end{center}

The set of control clones 
$C = \braced{c_1,c_2,...,c_8}$ and the resulting bipartite graph $G$
are shown below ($G$ is represented by its $C\times P$
adjacency matrix).

\noindent
\begin{center}
\setlength{\tabcolsep}{3pt}
\setlength{\columnsep}{0mm}
\begin{tabular}{llr|lccccccc} 
    & &\ &\ & 
			$p_1$  & $p_2$ & $p_3$ & $p_4$ & $p_5$ & $p_6$ & $p_7$
	\\ \hline
$c_1$ &
	\dnaseq{ATTGAACG\underline{CTGGC}GGCAGGCCTAACACATGCAAGTCGGACGGTAG}
			& &  &
				1 & 0 & 0 & 1 & 0 & 1 & 0
				 \\
$c_2$ &
	\dnaseq{GACGAACA\underline{GCCAG}GGCGTGCTTCGGCGATGCAAGTCGAGCGCTAA}
				& &  &
				1 & 0 & 1 & 0 & 1 & 0 & 0
				 \\
$c_3$ &
	\dnaseq{ATTTTACG\underline{CTGGC}GGCAGGCCTAACACATGCAAGTCGAAAAGTAG}
			& &  &
				1 & 0 & 0 & 1 & 0 & 1 & 0 
				\\
$c_4$ &
	\dnaseq{ACGCTAGCGGGATGCTTTACACATGCAAGTCGAACGGCAATACAT}
			& &  &
				0 & 1 & 0 & 0 & 1 & 0 & 1
	\\
$c_5$ & 
	\dnaseq{ACGAACG\underline{CTGGC}GGCGTGCCTAATACATGCAAGTCGAACGCTTCT}
			& &  &
				1 & 1 & 1 & 1 & 0 & 1 & 1 
	\\
$c_6$ &
	\dnaseq{ACGAACG\underline{GCCAG}GGCGTGGATTAGGCATGCAACGGCGACGCTGGA}
			& & &
				1 & 0 & 1 & 1 & 0 & 1 & 0 
	\\
$c_7$ &
	\dnaseq{GATGAACGCTAGCGGCAGGCTTAATACATGCAAGTCGAACGGCAG}
			& & & 
				0 & 1 & 0 & 0 & 1 & 0 & 1 
	\\
$c_8$ &
	\dnaseq{GACGAACG\underline{CTGGC}GGCGTGCTTAACACATGCAAGTCGAACGGAAA}
			& & &
				1 & 0 & 1 & 1 & 0 & 0 & 0
\end{tabular}

\end{center}

\smallskip

In $G$, we have an edge $(c_i,p_j)$ if $p_j$ matches $c_i$, that is,
$p_j$ or its reverse complement appears in $c_i$. 
For example, $p_1$ appears in $c_1$, $c_3$, $c_5$ and $c_8$,
and its reverse complement {\tt GCCAG}
appears in $c_2$ and $c_6$.
(These occurrences of $p_1$ are underlined.)
There is no edge $(c_4,p_1)$ and $(c_7,p_1)$ since
$c_4$ and $c_7$ do not contain either $p_1$ or $p_1$'s reverse complement.

Now suppose that we want to select $s = 6$ control clones from $C$. The
probe degree sequence with respect to $D_1 = \braced{c_1,c_2,c_3,c_5,c_6,c_8}$
is $(6,1,4,5,1,4,1)$, while the probe degree sequence with respect to
$D_2 = \braced{c_2,c_4,c_5,c_6,c_7,c_8}$ is $(4,3,4,3,3,2,3)$. Thus
$D_2$ would be considered a better set of control clones, since more
degrees are closer to $s/2 = 3$.

\medskip

Generally, as mentioned earlier,
our goal is to find a set $D\subseteq C$
of cardinality $s$ such that,
for each $p\in P$, $\deg_D(p)$ is close to $s/2$.
Several objective functions can be studied.
To measure the deviation from the perfectly
balanced cover, for a given probe $p$, we can compute either
$\min\braced{\deg_D(p),s-\deg_D(p)}$ or
$|\deg_D(p)-s/2|$. The objective function can
be obtained by considering the average of these
values or the worst case over all probes.
This gives rise to four objective functions:
\begin{eqnarray*}
\mbox{\rm maximize}\quad
\Cmin{}(D) &=& \min_{p\in P}\min\braced{\deg_D(p),s-\deg_D(p)}
                \\
\mbox{\rm maximize}\quad
\Cavg{}(D) &=& \avgn \sum_{p\in P}\min\braced{\deg_D(p),s-\deg_D(p)}
                \\
\mbox{\rm minimize}\quad
\Dmax{}(D) &=& \max_{p\in P}|\deg_D(p)-s/2|
                \\
\mbox{\rm minimize}\quad
\Davg{}(D) &=& \avgn \sum_{p\in P}|\deg_D(p)-s/2|
\end{eqnarray*}
where each function needs to be optimized over all choices
of $D$. There are certain relations among these functions,
for an instance, maximizing $\Cmin{}(D)$ and $\Cavg{}(D)$
is equivalent to minimizing $\Dmax{}(D)$ and $\Davg{}(D)$,
respectively, since $\Dmax{}(D) = s/2 - \Cmin{}(D)$,
and $\Davg{}(D) = s/2 - \Cavg{}(D)$.
Throughout the paper, the four
optimization problems corresponding to these
functions will be denoted by
{\BCPCmin}, {\BCPCavg}, {\BCPDmax}, and {\BCPDavg}.

Let $\angled{G,s}$ be an instance of Balanced
Covering. By $\Cmin{\ast}(G,s) = \max_D \Cmin{}(D)$ we denote
the optimal value of $\Cmin{}(D)$.
If $\calA$ is an algorithm for {\BCPCmin},
then $\Cmin{\calA}(G,s)$ denotes the
value computed by $\calA$ on input $\angled{G,s}$.
We use similar notations, $\Cavg{\ast}(G,s)$,
$\Cavg{\calA}(G,s)$, etc.,
for all the other objective functions introduced above.

%%%%%%%%%%%%%%%

\myparagraph{Results.}
In this paper we show several analytical and
experimental results on Balanced Covering.
In Section~\ref{sec: np-complete}
we prove that all versions of Balanced Covering
are {\NP}-hard. In particular, it is
{\NP}-complete to decide whether there is
a perfectly balanced cover with $s$ clones,
as well as to decide whether
there is a size-$s$ cover
where each probe is covered by at least one
but not all clones.
These results immediately imply
that (unless {\PP}={\NP}), there are no
polynomial-time approximation algorithms for
{\BCPDmax}, {\BCPDavg}, and {\BCPCmin}.

Stronger hardness-of-approximation results are shown in
Section~\ref{sec: hardness app}. For example,
for {\BCPDavg}, we show that approximating the optimum is
hard even if we allow randomization and an additive term in the
performance bound.
More specifically, we prove that, unless $\RP = \NP$, there is
no randomized polynomial-time algorithm $\calA$ that
for some constants $\alpha > 0$ and $0 < \epsilon < 1$ satisfies
$\Exp[\Davg{\calA}] \le \alpha \Davg{\ast} + \beta \avgn (mn)^{1-\epsilon}$.
($\RP$ is the class of decision problems that can be solved in
randomized polynomial time with one-side error.)
For {\BCPCmin}, we show that there is no polynomial-time algorithm
that computes a solution with the objective value at least
$\Cmin{\ast} - (1-\epsilon)\ln n$ or $\epsilon \Cmin{\ast}$,
unless {\NP} has slightly superpolynomial-time algorithms.
Our results on hardness of approximation are
summarized in Table~\ref{tab1}.

Then, in Section~\ref{sec: algorithms},
we propose a polynomial-time randomized
rounding algorithm {\RCM} for {\BCPCmin}.
The algorithm solves the linear relaxation of
the integer program for {\BCPCmin}, and
then uses the solution to randomly pick
an approximately balanced cover.
We show that,
with probability at least $\half$, {\RCM}'s solution
has objective value at least
$\Cmin{\ast}- O(\sqrt{\Cmin{\ast} \ln n} + \sqrts)$.

Algorithm {\RCM} performs well for input
instances where the optimum is relatively
large, but its performance bound
can be improved further for instances where the
optimum is small compared to $s$.
In Section~\ref{sec: RCMtwo algorithm}, we present
another algorithm called {\RCMtwo} that,
with probability at least $\half$, computes a solution
with objective value at least
$\Cmin{\ast} - O(\sqrt{\Cmin{\ast} \ln n})$.
(Although the asymptotic
approximation bound of {\RCM} is not as
good as that of {\RCMtwo}, we include {\RCM} in the
paper because, according to our experiments
discussed in Section~\ref{sec: experiments},
it outperforms {\RCMtwo} in practice.)

We also study problems {\BCPDmax} and {\BCPCavg}, for which
we develop some polynomial-time
randomized rounding algorithms ({\RDM} for {\BCPDmax}, and
two algorithms {\RCA} and {\RCAtwo} for {\BCPCavg}.)
These results are summarized in Table~\ref{tab1}.

\begin{table}
%% increase table row spacing, adjust to taste
%\renewcommand{\arraystretch}{1.3}
\caption{Hardness results and algorithms}
\label{tab1}
\begin{center}
\begin{minipage}{\textwidth}
\setlength{\tabcolsep}{4pt}
\setlength{\columnsep}{0mm}
{\footnotesize
\begin{tabular}{|l||p{5cm}|p{5.3cm}|}
  \hline
   & Hardness Results & Randomized Polynomial-Time Alg. \\
  \hline\hline
        {\BCPCmin}
	&
        No polynomial-time algorithm $\calA$ satisfies 
	$\Cmin{\calA} \ge \Cmin{\ast} - (1-\epsilon)\ln n$
	or $\Cmin{\calA} \ge \epsilon \Cmin{\ast}$,
        unless $\NP \subseteq \DTIME(n^{O(\log\log n)})$.  
	&
        Algorithm {\RCMtwo} s.t.,
        $\Cmin{\malyRCMtwo} \ge \Cmin{\ast} - O\left(\sqrt{\Cmin{\ast} \ln n}\right)$,
	with probability at least $\half$.
	\\
  \hline
        {\BCPCavg}
	&
        No randomized polynomial-time algorithm $\calA$ satisfies
        $\Exp[\Cavg{\calA}] \ge \Cavg{\ast} - \beta \avgn (mn)^{1-\epsilon}$, unless $\RP = \NP$.
        &
        Algorithm {\RCAtwo} s.t.,
        $\Exp[\Cavg{\malyRCAtwo}] \ge \Cavg{\ast} - O\left(\sqrt{\Cavg{\ast}} \right)$.
	\\
  \hline
        {\BCPDmax}
	&
        No polynomial-time algorithm that $\calA$
        satisfies $\Dmax{\calA} \le \Dmax{\ast} + (1-\epsilon)\ln n$,
        unless $\NP \subseteq \DTIME(n^{O(\log\log n)})$.
        &
        Algorithm {\RDM} s.t.,
        $\Dmax{\malyRDM} \le \Dmax{\ast} + O\left( \sqrt{s\ln n} \right)$,
	with probability at least $\half$
	\\
  \hline
        {\BCPDavg}
	&
        No randomized polynomial-time algorithm $\calA$ satisfies
        $\Exp[\Davg{\calA}] \le \alpha \Davg{\ast} + \beta \avgn (mn)^{1-\epsilon}$, unless $\RP = \NP$.
	&
        \\
  \hline
\multicolumn{3}{l}{$\alpha, \beta > 0$ and $0< \epsilon <1$ are any constants.}
\end{tabular}
}
\end{minipage}
\end{center}
\end{table}

In Section~\ref{sec: experiments}, we present the results of our
experimental studies, where we tested algorithms {\RCM}, {\RCMtwo} and {\RDM}
on both synthetic and real data sets.
According to this study,
solutions found by these algorithms
are very close to the optimal solution of
their corresponding linear program,
especially on real data sets. For example,
in $92.8\%$ of our real data sets,
{\RCM} found the solution with value at least $97\%$
of the solution from the linear program.

Algorithm {\RCM} has been implemented and
is publicly available at the OFRG website
as part of the CloneTools software package, see
{\tt http://algorithms.cs.ucr.edu/OFRG/}.

%%%%%%%%%%%%%%%%%%%

\myparagraph{Relation to other work.}
We are not aware of any other work on the Balanced Covering
problem studied in this paper.

Note that OFRG differs from other
array-based analysis approaches that, typically,
involve a single microarray experiment where one
clone of interest is hybridized against a collection of arrayed
probes, each targeting a specific sequence.
These experiments include control clones as well, but these
control clones are used
to test whether they bind as predicted to particular
microarray probes
(see \cite{Johan00,Yu03}, for example).
In contrast, OFRG uses a small set of
probes (roughly 30-50) to coordinately distinguish a much
larger set of sequences (for example, all bacterial
rRNA genes). Each probe is used in one hybridization
experiment, and the unknown DNA clone sequences
are immobilized on the array.

%%%%%%%%%%%%%%%%%%%%%%%%%%%%%%%%%%%%%%%%%%%%%%%%%%%%%%%%
%%%%%%%%%%%%%%%%%%%%%%%%%%%%%%%%%%%%%%%%%%%%%%%%%%%%%%%%
%%%%%%%%%%%%%%%%%%%%%%%%%%%%%%%%%%%%%%%%%%%%%%%%%%%%%%%%

\section{NP-Completeness}
\label{sec: np-complete}

We first show that all four versions of
Balanced Covering studied in this paper are {\NP}-hard.
In fact, we give two proofs of
{\NP}-hardness, as each will lead to
different results on hardness of approximation
in the next section.

Given a
bipartite graph $G = (C,P,E)$ and an even integer $s$,
define a \emph{perfectly balanced cover} in $G$ to be
a subset $D\subseteq C$
with $|D|=s$ such that $\deg_D(p) = s/2$ for each $p\in P$.
Similarly, we define a \emph{size-$s$ cover} to be a subset
$D\subseteq C$ with $|D|=s$ such that
$1\le \deg_D(p) \le s-1$ for each $p\in P$.

%%%%%%%%%%%%%

\begin{theorem}\label{thm: np-hard}
The following decision problem is {\NP}-complete:
``Given a bipartite graph $G = (C,P,E)$ and an even integer $s$,
is there a perfectly balanced cover in $G$?''
Consequently, {\BCPCmin}, {\BCPCavg}, {\BCPDmax} and {\BCPDavg}
are {\NP}-hard.
\end{theorem}

\begin{proof}
The proof is by a polynomial-time
reduction from {\XthreeC} (Exact Cover by 3-Sets),
which is known to be {\NP}-complete (see \cite{Garey79}, for
example).  The instance of {\XthreeC} consists of a finite set
$X$ of $3m$ items, and a collection $T$ of $n$ $3$-element
subsets of $X$ that we refer to as \emph{triples}.
We assume that $n\ge m\ge 2$. The objective is to
determine whether $T$ contains an \emph{exact cover of $X$},
that is a sub-collection $T'\subseteq T$ such that
every element of $X$ occurs in exactly one triple in $T'$.

The reduction is defined as follows.
Given an instance $\angled{X,T}$ of {\XthreeC} above, we construct
an instance $\angled{G=(T\cup W,X,E),s}$ of Balanced Covering,
where $W$ is a set that contains $m-2$ new vertices.
For $t\in T$ and $x\in X$,
we create an edge $(t,x)\in E$ if $x\in t$.
Further, we create all edges $(w,x)\in E$ for $x\in X$
and $w\in W$.
This defines the bipartite graph $G$. We let $s=2m-2$.

It remains to show that this construction is
correct, namely that $\angled{X,T}$ has an exact cover
iff $\angled{G,s}$ has a perfectly balanced cover.

%%%%%%%%

$(\Rightarrow)$
If $\angled{X,T}$ has an exact cover $T'$, we claim that $D=T'\cup W$ is
a perfectly balanced cover for $\angled{G,s}$.
To justify this, note first that $|T'|=m$ and $|W|=m-2$, and
thus $|D|=2m-2=s$.
Further, each vertex $x\in X$ has exactly one neighbor
in $T'$ and $m-2$ neighbors in $W$, so $x$ has
$m-1 = s/2$ neighbors in $D$, as required.

%%%%%%%%%

$(\Leftarrow)$
Suppose now that $\angled{G,s}$ has a perfectly balanced
cover $D\subseteq T\cup W$. Denote $W' = D\cap W$, $k = |W'|$, and
$T' = D\cap T$. We claim that $T'$ is an exact cover of $X$.

We first show that $D$ must contain all vertices in $W$.
We count the edges between $D$ and $X$.
There are $3km$ edges between $W'$ and $X$, since
each vertex in $W'$ is connected to all $3m$ vertices
in $X$. There are $3(s-k)$ edges between
$T'$ and $X$, since each vertex
in $T$ has degree $3$.
On the other hand, there must be $3m(m-1)$ edges
between $X$ and $D$, since each vertex in $X$
must be connected to exactly $s/2 = m-1$
vertices in $D$. Together, this yields
$3(2m-2-k)+3km \;=\; 3m(m-1)$. Solving this equation, we get
$k=m-2$, which means that $W' = W$.

Since $W'=W$, $T'$ must contain exactly $s-k = m$ vertices.
Each vertex $x\in X$ is adjacent to all vertices
in $W$, so it has exactly $s/2-(m-2)=1$ neighbor
in $T'$. This means that $T'$ is an exact cover of $X$, as claimed.
\end{proof}

%%%%%%%%%%%%%%%%%%%

Next we prove that it is {\NP}-complete to decide whether
there is a size-$s$ cover, where each probe in $P$ is covered
by at least one but not all clones from the cover.

\begin{theorem}\label{thm: size-s np-hard}
The following decision problem is {\NP}-complete:
``Given a bipartite graph $G = (C,P,E)$ and an integer $s$,
is there a size-$s$ cover in $G$?''
\end{theorem}

\begin{proof}
The proof is by a polynomial-time
reduction from the {\NP}-complete problem
Set Cover (see \cite{Garey79}). Given an instance of Set Cover
$\angled{Q,X,b}$, where $Q$ is a collection of subsets
over universe $X$, the query is whether
there is a set cover of size $b$ for $X$, that is a sub-collection
$Q'\subseteq Q$ with $|Q'|=b$ such that $\bigcup Q' = X$.

The reduction is defined as follows.
Given an instance $\angled{Q,X,b}$ of Set Cover, we
construct an instance $\angled{G=(Q\cup \{q_0\},X\cup \{x_0\},E),s}$,
where $q_0$ and $x_0$ are two new vertices. For
$q\in Q$ and $x\in X$, we create an edge $(q,x)\in E$ if $x\in q$.
We also create all edges $(q,x_0)\in E$ for $q\in Q$.
This defines the bipartite graph $G$. We let $s=b+1$.

We now justify the correctness of the construction by
showing that $\angled{G,s}$ has a size-$s$ cover iff
$\angled{Q,X,b}$ has a set cover with size $b$.

%%%%%%%%

$(\Rightarrow)$
If $\angled{Q,X,b}$ has a set cover $Q'$ of size $b$,
it is clear that $D = Q'\cup \{q_0\}$ is a size-$s$ cover
for $\angled{G,s}$ since
each vertex in $X\cup \{x_0\}$ is adjacent to at least
one element from $Q'$, and not adjacent to $q_0\in D$.

%%%%%%%%%

$(\Leftarrow)$
Suppose now that $\angled{G,s}$ has a size-$s$ cover
$D$. We denote $Q'=D\cap Q$.
Every $x\in X$ must be adjacent to at least
one vertex in $Q'$ since there is no $x$
adjacent to $q_0$.
Thus $Q'$ is a set cover of $X$ of size $b$.
\end{proof}

%%%%%%%%%%%%%%%%%%%%%%%%%%%%%%%%%%%%%%%%%%%%%%%%%%%%%%%%
%%%%%%%%%%%%%%%%%%%%%%%%%%%%%%%%%%%%%%%%%%%%%%%%%%%%%%%%
%%%%%%%%%%%%%%%%%%%%%%%%%%%%%%%%%%%%%%%%%%%%%%%%%%%%%%%%

\section{Hardness of Approximation}
\label{sec: hardness app}

\myparagraph{Approximation of {\BCPDavg} and {\BCPCavg}.}
Now we prove that approximating {\BCPDavg} and {\BCPCavg}
is hard.  Theorem~\ref{thm: np-hard} immediately implies
that {\BCPDavg} (as well as {\BCPDmax}) cannot be
efficiently approximated with any finite ratio. We show
that even if we allow an additive term in the approximation
bound and randomization, achieving finite ratio for
{\BCPDavg} is still {\NP}-hard. For {\BCPCavg} we show that
it is hard to be approximated with the bound $\Cavg{\ast} -
\beta \avgn (nm)^{1-\epsilon}$, where $\beta > 0$, $0<
\epsilon <1$ and $m = |C|$, $n = |P|$.

Let $\angled{G,s}$ be an instance of Balanced Covering.
Given an algorithm $\calA$ for {\BCPDavg},
recall that by $\Davg{\calA}(G,s)$ we denote
the value of the objective function
computed by $\calA$,
that is
\begin{equation*}
\Davg{\calA}(G,s) \;=\; \avgn \sum_{p\in P}|\deg_D(p)-s/2|,
\end{equation*}
where $D\subseteq C$ is the set computed by $\calA$.
Similarly, given an algorithm $\calA$ for {\BCPCavg},
$\Cavg{\calA}(G,s)$ is
the value of the objective function
computed by $\calA$ for {\BCPCavg}, that is
\begin{equation*}
\Cavg{\calA}(G,s) \;=\; \avgn \sum_{p\in P}\min\braced{\deg_D(p),s-\deg_D(p)}.
\end{equation*}
Recall that by $\Davg{\ast}(G,s)$ and $\Cavg{\ast}(G,s)$
we denote the optimal value for {\BCPDavg} and {\BCPCavg}, respectively.

Recall that the class {\RP} (randomized polynomial time)
is the complexity class of decision problems $\calP$ which
have polynomial-time probabilistic Turing machines
$M$ such that, for each input $I$,
(i)
if $I\in \calP$ then $M$ accepts $I$ with
probability at least $\half$, and
(ii)
if $I\notin \calP$ then $M$
rejects $I$ with probability $1$.
It is still open whether ${\RP} = {\NP}$.

%%%%%%%%%%%%%%%%%%%%%%%%%%%%%%%%%

\begin{theorem}\label{thm: hardness randomized}
Let $\alpha, \beta >0$ and $0 < \epsilon < 1$ be any
constants. If ${\RP}\neq {\NP}$ then there is no randomized
polynomial-time algorithm $\calA$ that

(a) for any instance $\angled{G,s}$ of {\BCPDavg} satisfies
\begin{equation}
\Exp[\Davg{\calA}(G,s)] \;\leq\;
         \alpha \cdot\Davg{\ast}(G,s)
                 + \beta \avgn (nm)^{1-\epsilon}, \quad \mbox{or}
                \label{eqn: bcpdavg rand approximation}
\end{equation}

(b) for any instance $\angled{G,s}$ of {\BCPCavg} satisfies
\begin{equation}
\Exp[\Cavg{\calA}(G,s)] \;\geq\;
         \Cavg{\ast}(G,s)
                 - \beta \avgn (nm)^{1-\epsilon}.
                \label{eqn: bcpcavg rand approximation}
\end{equation}
\end{theorem}

\begin{proof}
We first prove part (a) of the theorem.
Suppose, towards contradiction, that
for some $\alpha$, $\beta$ and $\epsilon$ there exists
a randomized polynomial-time
algorithm $\calA$ that satisfies (\ref{eqn: bcpdavg rand approximation}).
We show that this would imply the existence of a randomized
polynomial-time algorithm that decides if there is a perfectly
balanced covering, contradicting Theorem~\ref{thm: np-hard}.

Given an instance $\angled{G,s}$ of {\BCPDavg}, where
$G=(C,P,E)$, convert it into another instance
$\angled{G^r,s}$ of {\BCPDavg}, where $G^r = (C,P',E')$ is
obtained by creating $r$ copies of each probe $p \in P$
(that is, with the same neighbors in $C$). Thus $|P'| =
rn$. We choose
$r = \ceiling{(2\beta m^{1-\epsilon} n^{1-\epsilon} )^{\frac{1}{\epsilon}}}+1$.
For this $r$, we have $2\beta m^{1-\epsilon} (nr)^{-\epsilon} < \avgn$.
Therefore the new instance $\angled{G^r,s}$ has the
following properties:

\begin{itemize}
\item
If $\angled{G,s}$ has a perfectly balanced cover (that is,
$\Davg{\ast}(G,s) = 0$) then $\Davg{\ast}(G^r,s) = 0$, and
therefore $2 \cdot \Exp[\Davg{\calA}(G^r,s)] \le 2\beta
m^{1-\epsilon} (nr)^{-\epsilon} < \avgn$. Using Markov's
inequality, this implies that $\Prob[\Davg{\calA}(G^r,s) <
\avgn] \ge \half$.
\item
if $\angled{G,s}$ does not have a perfectly balanced cover
(that is, $\Davg{\ast}(G,s) \ge \avgn$) then
$\Davg{\calA}(G^r,s)
         \ge \Davg{\ast}(G^r,s)
         = \Davg{\ast}(G,s)
         \ge \avgn$,
with probability $1$.
\end{itemize}
Since $G^r$ can be computed from $G$ in polynomial time,
from $\calA$ we could obtain a randomized
polynomial-time algorithm that determines
the existence of a perfectly balanced cover --
a problem that is {\NP}-complete, according to
Theorem~\ref{thm: np-hard}. The part (a) of the theorem follows.

Part (b) follows directly from part (a) of the theorem
and the fact that $\Cavg{\ast}(G,s) = s/2 - \Davg{\ast}(G,s)$ and
$\Cavg{}(H) = s/2 - \Davg{}(H)$ for any solution $H$
for instance $\angled{G,s}$ of {\BCPCavg} and {\BCPDavg}.
\end{proof}

Using an argument very similar to the proof of Theorem~\ref{thm: hardness randomized},
one can show that, unless ${\PP} = {\NP}$,
there is no deterministic polynomial-time algorithm that satisfies bounds
analogous to those in Theorem~\ref{thm: hardness randomized}.

%%%%%%%%%%%%

\myparagraph{Approximation of {\BCPCmin} and {\BCPDmax}.}
Next we show that {\BCPCmin} cannot be approximated
efficiently with the objective value at least $\epsilon \Cmin{\ast}(G,s)$
or $\Cmin{\ast}(G,s)- O(\ln n)$,
unless {\NP} has slightly superpolynomial time algorithms.
As a result, {\BCPDmax} cannot be approximated
efficiently with the objective value at most
$\Dmax{\ast}(G,s) + O(\ln n)$.
Recall that for a given instance $\angled{G,s}$, we
denote by $\Cmin{\ast}(G,s)$ and $\Dmax{\ast}(G,s)$
the optimal value of $\Cmin{}(G,s)$
and $\Dmax{}(G,s)$, respectively.
Similarly, $\Cmin{\calA}(G,s)$ and $\Dmax{\calA}(G,s)$
are the values of the objective function computed by
an algorithm $\calA$ for {\BCPCmin} or {\BCPDmax},
respectively, on an instance $\angled{G,s}$.

\begin{theorem}\label{thm: hardness BCPCmin}
Unless $\NP \subseteq \DTIME(n^{O(\log \log n)})$, then

(a)
there is no polynomial-time algorithm $\calA$ for {\BCPCmin}
that, for some $0< \epsilon <1$, for any instance $\angled{G,s}$,
satisfies
\begin{equation}
\Cmin{\calA}(G,s)
\,\ge\,
\epsilon \, \Cmin{\ast}(G,s), \quad \mbox{and}
    \label{eqn: bcpcmin approximation 1}
\end{equation}

(b)
there is no polynomial-time algorithm $\calA$ for {\BCPCmin}
that, for some $0< \epsilon <1$, for any instance $\angled{G,s}$,
satisfies
\begin{equation}
\Cmin{\calA}(G,s)
\,\ge\,
\Cmin{\ast}(G,s) - (1-\epsilon) \ln n.
    \label{eqn: bcpcmin approximation 2}
\end{equation}

\end{theorem}

\begin{proof}
We first prove part (a) of the theorem. Suppose, towards contradiction,
that there exists a polynomial-time algorithm $\calA$ that
satisfies (\ref{eqn: bcpcmin approximation 1}). We show that this
would imply the existence of a polynomial-time $((1 - \Omega
(\epsilon) )\ln n)$-approximation algorithm $\calB_1$ for the Set
Cover problem, which would imply in turn that problems in {\NP} have
$n^{O(\log\log n)}$-time deterministic algorithms \cite{Feige98Threshold}.

Algorithm $\calB_1$ works as follows. Given an instance
$\angled{Q,X}$ of Set Cover, where $|X|=n$ and $Q$ is a collection of sets
over $X$, the algorithm $\calB_1$ first reduces $\angled{Q,X}$
to an instance $\angled{G=(T\cup W, P, E),s}$ of $\BCPCmin$, where
$P=X\cup\{x_0\}$, $T$ contains $k=\floor{\frac{\ln n}{2} }$
vertices $q_1, q_2,..., q_k$ for each set $q\in Q$, and $W$ is a set
containing $k$ new vertices. For each $q\in Q$ and $i = 1,2,...,k$,
we create an edge $(q_i,x_0)\in E$ and edges $(q_i,x)\in E$ for
each $x\in q$. This defines the bipartite graph $G$.
Let $b$ represent the size of the minimum set cover of $X$.
We now assume that, without loss of generality, algorithm $\calB_1$
knows the value of $b$. Otherwise, $\calB_1$ can simply
try each $b\in\{1,2,...,n\}$,
and choose the smallest set cover. We now let $s=kb+k$.

Next $\calB_1$ calls algorithm $\calA$ on input $\angled{G,s}$
to get a balanced cover $H$ for $\angled{G,s}$, and outputs the
collection of sets $H' = \{q : (\exists i) q_i \in H\}$ as a
set cover of $X$.

To prove that $\calB_1$ is a
$((1 - \Omega (\epsilon) )\ln n)$-approximation algorithm
for Set Cover, we now show that $H'$ is a set
cover of $X$ and $|H'| \leq (1-\frac{\epsilon}{2})\ln (n) b$.

Assuming that $\angled{Q,X}$ has a set cover $Q'$ of size $b$, we
first claim that $\Cmin{\ast}(G,s) \ge k$. To justify this, from
$Q'$, we build the balanced cover $D=\{q_i: q\in Q'\}\cup W$.
Obviously, $|D| = kb+k = s$. For each $x\in X$, the $k$ copies of
$Q'$ ensure that $\deg_D(x)\ge k$, while the $k$ vertices in $W$
ensure that $\deg_D(x) \le s-k$. Our claim implies that algorithm
$\calA$ on input $\angled{G,s}$ will find a balanced cover $H$ with
objective function value at least $\epsilon\,\Cmin{\ast}(G,s) \ge
\epsilon k$. We have $|H\cap W|\ge \epsilon k$, because $x_0\in P$
is adjacent to every vertex in $T$ and $H$ has at least $\epsilon k$
vertices not adjacent to $x_0$. Therefore $|H\cap T|\le s - \epsilon
k = kb + (1-\epsilon)k$. Thus, since each $x\in P$ is adjacent to at
least one vertex in $H$ (in fact, at least $\epsilon k$), $H'$ forms
a set cover of $X$ of size at most $
  kb + (1-\epsilon)k
  \,\le\,
  (2-\epsilon)kb
  \,\le\,
  (1 - \frac{\epsilon}{2})\ln (n) b,
$
as claimed.

The algorithm $\calB_1$ clearly runs in polynomial time, and
is a $((1 - \Omega (\epsilon) )\ln n)$-approximation algorithm for the
Set Cover problem. Thus the part (a) of the theorem follows.

Next we prove the part (b) of the theorem. Suppose, towards
contradiction, that there exists a polynomial-time algorithm $\calA$ that
satisfies (\ref{eqn: bcpcmin approximation 2}).
As in part (a), we will prove that this would imply the
existence of a polynomial-time
$((1-\Omega(\epsilon))\ln n)$-approximation algorithm
$\calB_2$ for the Set Cover problem.

$\calB_2$ works like algorithm $\calB_1$ described previously
except we let $k=\ceiling{(1-\epsilon)\ln n + 1}$ this time.

Assuming that $\angled{Q,X}$ has a set cover $Q'$ of size $b$,
an argument similar to the proof of part (a) shows that algorithm $\calA$
on input $\angled{G,s}$ will find a balanced cover $H$ with
objective function value at least
\[
  \Cmin{\ast}(G,s) - (1-\epsilon)\ln n
  \,\ge\,
  k - (1-\epsilon)\ln n
  \,\ge\,
  1.
\]
Also, $H'$ forms a set cover of $X$ of size at most $s$.
We now assume that, without loss of generality, $b \geq \frac{2}{\epsilon}$,
because otherwise the Set Cover problem can be solved in
polynomial time $O(|X|^2 \cdot |Q|^{\frac{2}{\epsilon}})$. Thus we get $
  |H'|
  \,\le\,
  kb + k
  \,\le\,
  kb + \frac{\epsilon}{2} kb
  \,\le\,
  (1+\frac{\epsilon}{2})(1-\epsilon)(\ln (n) + \frac{2}{1-\epsilon} )b
  \,\le\,
  (1-\Omega(\epsilon)) \ln(n) b,
$
as claimed.

The algorithm $\calB_2$ clearly runs in polynomial time,
and is a $((1-\Omega(\epsilon))\ln n)$-approximation algorithm for the Set
Cover problem. Thus the theorem follows.
\end{proof}

As a corollary, we also get an approximation hardness result for {\BCPDmax}.

\begin{corollary}\label{cor: hardness BCPDmax}
Unless $\NP \subseteq \DTIME(n^{O(\log \log n)})$,
there is no polynomial-time algorithm $\calA$
for {\BCPDmax} that,
for some $0< \epsilon <1$ and
for any instance $\angled{G,s}$, satisfies
\begin{equation}
\Dmax{\calA}(G,s)
\,\le\,
\Dmax{\ast}(G,s) + (1-\epsilon) \ln n.
    \label{eqn: bcpdmax approximation}
\end{equation}
\end{corollary}

\begin{proof}
The corollary follows directly from part (b) of Theorem~\ref{thm: hardness BCPCmin}
and the fact that $\Dmax{\ast}(G,s) = s/2 - \Cmin{\ast}(G,s)$ and
$\Dmax{}(H) = s/2 - \Cmin{}(H)$ for any solution $H$
for instance $\angled{G,s}$ of {\BCPDmax} and {\BCPCmin}.
\end{proof}

%%%%%%%%%%%%%%%%%%%%%%%%%%%%%%%%%%%%%%%%%%%%%%%%%%%%%%%%
%%%%%%%%%%%%%%%%%%%%%%%%%%%%%%%%%%%%%%%%%%%%%%%%%%%%%%%%
%%%%%%%%%%%%%%%%%%%%%%%%%%%%%%%%%%%%%%%%%%%%%%%%%%%%%%%%

\section{Approximation Algorithms and Analysis}
\label{sec: algorithms}

In this section we present several randomized algorithms for
different versions of Balanced Covering.
We give two algorithms {\RCM} and {\RCMtwo} for {\BCPCmin}, algorithm {\RDM}
for {\BCPDmax}, and two algorithms {\RCA} and {\RCAtwo} for {\BCPCavg}.

All algorithms are based on randomized rounding.
We first solve a linear relaxation LP of the integer program ILP for
Balanced Covering, and then use the fractional solution as probabilities to
randomly choose the integral solutions.

Let $\starx_1,...,\starx_n$, where $0\le \starx_i\le 1$
for each $i$, be the optimum solution of LP and $\starz$ the corresponding
optimum value of the objective function.
We choose $X_i =1$ with probability $\starx_i$ and $0$ otherwise, which
gives us a ``provisional'' integral
solution $X_1,...,X_n$ with objective value $Z$. Since the
expectation of $Z$ is equal to $\starz$ and the random variables $X_i$ are
independent, we can apply the Chernoff bound to show that with
high probability the value of $Z$ is close to $\starz$ (and thus
also approximates well the optimum of ILP).
If $Z$ is not feasible, we adjust the values of a sufficient number $L$ of
the variables $X_i$ obtaining a final
feasible solution whose value $\tildeZ$ differs
from $Z$ by at most $L$. Applying the Chernoff bound again, we get an
estimate on $L$, and combining it with the bound on $Z$
we obtain a bound on $\tildeZ$.

For some objective functions we refine this approach further, by
adjusting the probability of setting $X_i$ to $1$, in order to reduce the
violation $L$ of the constraints. This modification improves
asymptotic performance bounds but -- as we show later in
Section~\ref{sec: experiments} -- it tends to degrade the
experimental performance on both random and real data sets.

%%%%%%%%%%%%%%%%%%%%%%%%%
\subsection{Algorithm {\RCM} for {\BCPCmin}}
\label{sec: RCM algorithm} Given $G = (C,P,E)$, let
$C=\{c_1, c_2,..., c_m\}$, $P=\{p_1, p_2,..., p_n\}$.
And let $A=[a_{ij}]$ be the
Boolean $m\times n$ adjacency matrix of $G$,
that is $a_{ij}= 1$ iff $(c_i, p_j)\in E$;
otherwise $a_{ij} = 0$. Then {\BCPCmin} is
equivalent to the following integer linear program {\minIP}:
\begin{eqnarray*}
\mbox{\rm maximize:}
    \quad z &&
                        \\
\mbox{\rm subject to:}
    \quad z \;&\le\;& \textstyle{\sum_{i=1}^m a_{ij} x_i} \quad\forall j=1,...,n
                        \\
    \quad z \;&\le\;& \textstyle{\sum_{i=1}^m (1-a_{ij}) x_i} \quad\forall j=1,...,n
                        \\
    \quad \textstyle{\sum_{i=1}^m x_i} \;&\le\;& s
                        \\
        x_i \;&\in\;& \braced{0,1} \quad\forall i=1,...,m
\end{eqnarray*}
The Boolean variables $x_i$ indicate whether the
corresponding $c_i \in C$ are selected or not.

\myparagraph{Algorithm {\RCM}.}
The algorithm first relaxes the
last constraint to $0\leq x_i \leq 1$ to obtain the linear program
{\minLP}, and then computes an optimal solution $\starx_i$,
$i=1,2,...,m$, of {\minLP}. Next, applying randomized rounding,
{\RCM} computes an integral solution $X_1,...,X_m$ by choosing
$X_i = 1$ with probability $\starx_i$ and $0$ otherwise. Note that this
solution may not be feasible since $\sum_{i=1}^m X_i$ may exceed
$s$. Let $L = \max \{\sum_{i=1}^m X_i - s, 0\}$. {\RCM} changes $L$
arbitrary variables $X_i = 1$ to $0$, obtaining a feasible solution
$\tildeX_1,...,\tildeX_m$.

%%%%%%%%%%%%%%%%%%%%%

\myparagraph{Analysis.} \label{sec: analysis bcpcmin}
We denote by
$\Cmin{\malyRCM}(G,s)$ or $\tildeZ$ the value of the objective
function computed by {\RCM} on input $\angled{G,s}$, that is
$\Cmin{\malyRCM}(G,s) =
    \tildeZ = \min_{j=1}^n \{\sum_{i=1}^m a_{ij}\tildeX_i,
        \sum_{i=1}^m (1-a_{ij})\tildeX_i\}$.

%%%%%%%%%%%%%%%%%%%%%%%%%%%%%%%%%%%%%%

\begin{lemma}\label{lem: app bound}
For any instance $\angled{G,s}$ of {\BCPCmin},
with probability at least $\half$,
\begin{equation}
\Cmin{\malyRCM}(G,s) \;\ge\;
    \Cmin{\ast}(G,s) -
    O\left(\sqrt{\Cmin{\ast}(G,s) \ln n}
    + \sqrts\right).
        \label{eqn: RCM app bound}
\end{equation}
\end{lemma}

\begin{proof}
Let $\starz = \min_{j=1}^n
    \{\sum_{i=1}^m a_{ij}\starx_i,
        \sum_{i=1}^m (1-a_{ij})\starx_i\}$
be the optimum solution of {\minLP}.
Let also
$Z = \min_{j=1}^n \{\sum_{i=1}^m a_{ij}X_i,
                \sum_{i=1}^m (1-a_{ij})X_i\}$.

The $\{X_i\}$ are independent Bernoulli random variables
with $\Exp[X_i]=\starx_i$. So, for each $j$,
$\Exp[\sum_{i=1}^m a_{ij} X_i]
    = \sum_{i=1}^m a_{ij} x^*_i \ge z^*$,
By a standard Chernoff bound, we get
\[
\Prob[\textstyle\sum_{i=1}^m a_{ij} X_i \le (1-\lambda) z^*]
    \;\le\; e^{-\lambda ^2 z^* / 2},
\]
where $0< \lambda \le 1$. Similarly, for all $j$,
\[
\Prob[\textstyle\sum_{i=1}^m (1-a_{ij}) X_i \le (1-\lambda) z^*]
    \;\le\; e^{-\lambda ^2 z^* / 2}.
\]
By the naive union bound,
the probability that any of the $2n$ above events happens is at most
$2ne^{-\lambda ^2 z^*/2}$. Hence we have
\begin{equation}
\Prob[Z \le (1-\lambda) z^*]
    \;\le\; 2ne^{-\lambda ^2 z^* / 2}.
    \label{eqn: 2n bound}
\end{equation}

Likewise, $\Exp[\sum_{i=1}^m X_i] = \sum_{i=1}^m \starx_i \le s$.
Thus by the Chernoff bound,
$\Prob[\textstyle{\sum_{i=1}^m X_i} \ge (1+\epsilon) s]
        \le e^{- \epsilon ^2 s/4}$, where $0< \epsilon \le 2e-1$.
Recalling $L = \max \{\sum_{i=1}^m X_i - s, 0 \}$, we have
\begin{equation}
\Prob[L \ge \delta \sqrts]
        \;\le\; e^{- \delta ^2/4},
        \label{eqn: constraint violation}
\end{equation}
where $0< \delta \le (2e-1)\sqrts$.

Since $\tildeZ \ge Z-L$, we get
$\Prob[\tildeZ \le (1 - \lambda) \starz - \delta \sqrts]
    \le \Prob[Z \le (1 - \lambda) \starz] + \Prob[L \ge \delta \sqrts]$.
Combining this with (\ref{eqn: 2n bound}) and
(\ref{eqn: constraint violation}), we have
\begin{equation}
\Prob[\tildeZ \le (1 - \lambda) \starz - \delta \sqrts]
        \;\le\; 2ne^{-\lambda ^2 z^* / 2} + e^{- \delta ^2/4}.
    \label{eqn: 2n bound and constraint violation}
\end{equation}

Suppose $\Cmin{\ast}(G,s) \ge 2\ln (8n)$. Then $\starz \ge 2\ln (8n)$
as well, because $\Cmin{\ast}(G,s)\le \starz$. Choosing
$\lambda=\sqrt{2\ln (8n) /\starz}$ and $\delta=\sqrt{4\ln 4}$,
from (\ref{eqn: 2n bound and constraint violation}), we get
\[
\Prob[\tildeZ \le \starz - \sqrt{2\ln (8n)\starz} - \sqrt{4\ln (4) s}]
\;\le\; \half.
\]
Since $\starz \ge \Cmin{\ast}(G,s)\ge 2\ln (8n)$, with
probability at least $\half$, we have
\begin{equation}
\tildeZ \ge \Cmin{\ast}(G,s) -
    \sqrt{2\ln (8n) \Cmin{\ast}(G,s)} - \sqrt{4\ln (4) s}.
    \label{eqn: RCM final inequ}
\end{equation}

Inequality (\ref{eqn: RCM final inequ}) is also
trivially true for  $\Cmin{\ast}(G,s) < 2\ln (8n)$.
Thus the lemma follows.
\end{proof}

%%%%%%%%%%%%%%%%%%%%%

\subsection{An Alternative Algorithm {\RCMtwo} for {\BCPCmin}}
        \label{sec: RCMtwo algorithm}

The performance bound for {\RCM} given in Section \ref{sec: RCM
algorithm} can be improved for instances where the optimum
is small compared to $s$. We now provide an alternative algorithm
{\RCMtwo}, which is identical to {\RCM} in all steps except for
the rounding scheme: choose $X_i=1$ with probability
$(1-\epsilon)\starx_i$, and $0$ otherwise, where $\epsilon =
\min\braced{ 2\sqrt{\ln(4n+2)/z^*}, 1}$.

\myparagraph{Analysis.} All notations are defined similarly
to those in Section \ref{sec: RCM algorithm}.

\begin{lemma}
For any instance $\angled{G,s}$ of {\BCPCmin},
with probability at least $\half$,
\begin{equation}
    \Cmin{\malyRCMtwo}(G,s)
    \,\ge\,
    \Cmin{\ast}(G,s) - O\left(\sqrt{\Cmin{\ast}(G,s) \ln n}\right).
    \label{eqn: RCMtwo app bound}
\end{equation}
\end{lemma}
\begin{proof}
 The $\{X_i\}$ are independent random variables
 with $\Exp [X_i] = (1-\epsilon)x^*_i$.
 By linearity of expectation,
 $\Exp [\sum_{i=1}^m X_i] \le \sum_{i=1}^m (1-\epsilon)x^*_i \le (1-\epsilon)s$.
 Thus, by the Chernoff bound,
 \[\textstyle
 \Prob[\sum_{i=1}^m X_i \ge s]
 \,\le\,
 \Prob[\sum_{i=1}^m X_i \ge (1+\epsilon)(1-\epsilon)s]
 \,\le\, e^{-\epsilon^2(1-\epsilon) s / 4}.
 \]
 As $z^*\le s/2$, we have $s/4\ge z^*/2$. The above bound implies
 \begin{equation}\textstyle
 \Prob[\sum_{i=1}^m X_i \ge s]
 \,\le\,
 e^{-\epsilon^2(1-\epsilon) \starz / 2}.
 \label{eqn: RCM2 s bound}
 \end{equation}
 Likewise, for each $j$,
 $\sum_{i=1}^m a_{ij} x^*_i \ge z^*$,
 so $\Exp [\sum_{i=1}^m a_{ij} X_i] \ge (1-\epsilon)z^*$.
 By the Chernoff bound,
 \begin{equation}\textstyle
 \Prob[\sum_{i=1}^m a_{ij} X_i \le (1-\epsilon)^2z^*]
 \,\le\, e^{-\epsilon^2(1-\epsilon) z^* / 2}.
 \label{eqn: RCM2 Z1 bound}
 \end{equation}
 Similarly, for  all $j$,
 \begin{equation}\textstyle
 \Prob[\sum_{i=1}^m (1-a_{ij}) X_i \le (1-\epsilon)^2 z^*]
 \,\le\, e^{-\epsilon^2(1-\epsilon) z^* / 2}.
 \label{eqn: RCM2 Z0 bound}
 \end{equation}
 Letting $L=\max \{\sum_{i=1}^m X_i -s, 0\}$,
 since $\tildeZ \ge Z-L$,
 we get $\Prob[\tildeZ \le (1 - \epsilon)^2 \starz - L]
    \le \Prob[Z \le (1 - \epsilon)^2 \starz] + \Prob[\sum_{i=1}^m X_i \ge s]$.
 Combining this with (\ref{eqn: RCM2 s bound}), (\ref{eqn: RCM2 Z1 bound})
 and (\ref{eqn: RCM2 Z0 bound}), we have
 \[
 \Prob[\tildeZ \le (1 - \epsilon)^2 \starz]
 \;\le\; (2n+1)e^{-\epsilon^2(1-\epsilon)z^*/2}.
 \]
 Since $(1-\epsilon)^2 \ge  1- 2\epsilon$,
 for $\epsilon < \half$, we get
 \begin{equation}
 \Prob[\tildeZ \le \starz - 4\sqrt{\ln(4n+2)\starz}]
 \;\le\; \half.
 \label{eqn: RCM2 tildeZ bound}
 \end{equation}
 The above bound is also trivially true for $\epsilon \ge \half$
 (that is, $\starz \le 16\ln (4n+2)$).
 Finally, suppose $\Cmin{\ast}(G,s) \ge 16\ln (4n+2)$.
 Since also $\Cmin{\ast}(G,s) \le \starz$,
 inequality (\ref{eqn: RCM2 tildeZ bound}) implies that
 with probability at least $\half$,
 \begin{equation}
 \tildeZ \ge \Cmin{\ast}(G,s) - 4\sqrt{\ln(4n+2)\Cmin{\ast}(G,s)},
    \label{eqn: RCM2 final inequ}
 \end{equation}
 Inequality (\ref{eqn: RCM2 final inequ}) is also trivially true
 for $\Cmin{\ast}(G,s) \le 16\ln (4n+2)$.
 Thus the lemma follows.
\end{proof}

We will show later in Section \ref{sec: experiments} that
{\RCMtwo} does not outperform {\RCM} in experimental analysis.
Therefore {\RCM} cannot be completely substituted by {\RCMtwo}.

%%%%%%%%%%%%%%%%%%%%%%%%%%

\subsection{Algorithm {\RDM} for {\BCPDmax}}
\label{sec: RDM algorithm} In this section we present
our randomized algorithm {\RDM} for {\BCPDmax}. Given $G = (C,P,E)$, let $A$ be
the Boolean $m\times n$ adjacency matrix of $G$, as in
Section~\ref{sec: RCM algorithm}. Then {\BCPDmax} is equivalent to the
following integer linear program {\maxIP}:
\begin{eqnarray*}
\mbox{\rm minimize:}
    \quad z &&
                        \\
\mbox{\rm subject to:}
    \quad z \;&\ge\;& \textstyle{\sum_{i=1}^m a_{ij} x_i} - s/2 \quad\forall j=1,...,n
                        \\
    \quad z \;&\ge\;& \textstyle{s/2 - \sum_{i=1}^m a_{ij} x_i} \quad\forall j=1,...,n
                        \\
    \quad \textstyle{\sum_{i=1}^m x_i} \;&=\;& s
                        \\
        x_i \;&\in\;& \braced{0,1} \quad\forall i=1,...,m
\end{eqnarray*}
The Boolean variables $x_i$ indicate whether the corresponding
$c_i \in C$ are selected or not.

\myparagraph{Algorithm {\RDM}.} The algorithm first relaxes
the last constraint to $0\leq x_i\leq 1$ to obtain the linear program
{\maxLP}, and then computes an optimal solution $\starx_i$,
$i=1,2,...,m$, of {\maxLP}. Next, applying randomized rounding like in
{\RCM}, {\RDM} computes an integral solution $X_1,...,X_m$
by choosing $X_i = 1$ with probability $\starx_i$ and $0$ otherwise. Note
that this solution may not be feasible since $\sum_{i=1}^m X_i$
may not be exactly $s$. Let $L = \sum_{i=1}^m X_i - s$. {\RDM} changes
$|L|$ arbitrary variables $X_i = 1$ to $0$ if $L
> 0$, and does the contrary if $L < 0$, obtaining a feasible
solution $\tildeX_1,...,\tildeX_m$.

%%%%%%%%%%%%%%%%%%%%%

\myparagraph{Analysis.} \label{sec: analysis bcpdmax} We denote by
$\starz$ and $\Dmax{\malyRDM}(G,s)$ (or $\tildeZ$) the value of the
objective function computed by {\maxLP} and {\RDM} for {\BCPDmax},
respectively. Namely, $\starz = \max_{j=1}^n
    | \sum_{i=1}^m a_{ij}\starx_i - s/2 |$,
and
$\Dmax{\malyRDM}(G,s) =
    \tildeZ = \max_{j=1}^n | \sum_{i=1}^m a_{ij}\tildeX_i - s/2 |$.

%%%%%%%%%%%%%%%%%%%%%%%%%%%%%%%%%%%%%%

\begin{lemma}\label{lem: RDM bound}
For any instance $\angled{G,s}$ of {\BCPDmax}, with probability at least $\half$,
\begin{equation}
\Dmax{\malyRDM}(G,s) \;\le\;
    \Dmax{\ast}(G,s) + O\left(\sqrt{s \ln n} \right).
        \label{eqn: RDM bound}
\end{equation}
\end{lemma}

\begin{proof}
We now assume that, without loss of
generality, $s\ge 9$, because otherwise $s$ is a constant then
(\ref{eqn: RDM bound}) will be trivially true.

Let $Z = \max_{j=1}^n | \sum_{i=1}^m a_{ij}X_i - s/2 |$.
For each $j$, define $\barz_j = \sum_{i=1}^m a_{ij}\starx_i$, and $\starz_j =
| \barz_j - s/2 |$. Similarly, define random variables
$\barZ_j = \sum_{i=1}^m a_{ij} X_i$ and $Z_j = | \barZ_j - s/2 |$. Thus $\starz
= \max_{j=1}^n \starz_j$ and $Z = \max_{j=1}^n Z_j$.

The $\{X_i\}$ are independent Bernoulli random variables
with $\Exp[X_i]=\starx_i$. So $\Exp[\barZ_j] = \barz_j$ for each $j$.
Applying a standard Chernoff bound, we get
$\Prob[|\barZ_j - \barz_j| \ge \epsilon \barz_j]
    \le 2e^{-\epsilon ^2 \barz_j /4}$,
for $0< \epsilon \le 1$. This and the triangle inequality imply
\begin{equation}
 \Prob[Z_j \ge \starz + \lambda \sqrts]
 \;\le\;
 \Prob[Z_j \ge \starz_j + \lambda \sqrt{\barz_j}]
 \;\le\;
 \Prob[|\barZ_j - \barz_j| \ge \lambda \sqrt{\barz_j}]
 \;\le\;
 2e^{-\lambda ^2 /4},
 \label{eqn: Zj z* diff}
\end{equation}
where $0< \lambda \le \sqrt{\barz_j}$. Since
$|\barz_j - s/2| \le \starz$, $\barz_j \ge s/2 - \starz$ for all $j$.
Hence (\ref{eqn: Zj z* diff}) also holds when
$0< \lambda \le \sqrt{s/2 - \starz}$.

By the naive union bound,
\begin{equation}
\Prob[Z \ge \starz + \lambda \sqrts]
    \;\le\; 2ne^{-\lambda ^2 /4}.
    \label{eqn: n bound}
\end{equation}

Likewise, $\Exp[\sum_{i=1}^m X_i] = \sum_{i=1}^m \starx_i = s$.
By the Chernoff bound,
$\Prob[ | \textstyle{\sum_{i=1}^m X_i} - s | \ge \epsilon s]
        \le 2e^{- \epsilon ^2 s/4}$,
for $0< \epsilon \le 1$. Thus we have
\begin{equation}
\Prob[ | \textstyle{\sum_{i=1}^m X_i} - s | \ge \delta \sqrts]
        \;\le\; 2e^{- \delta ^2 /4},
        \label{eqn: constraint violation maxLP}
\end{equation}
where $0< \delta \le \sqrts$.

Since $\tildeZ \le Z + |\sum_{i=1}^m X_i - s|$, we get
$
\Prob[\tildeZ \ge \starz + \lambda \sqrts + \delta \sqrts]
  \le \Prob[Z \ge \starz + \lambda \sqrts] +
    \Prob[|\sum_{i=1}^m X_i - s| \ge \delta \sqrts]
$.
Combining this with (\ref{eqn: n bound}) and
(\ref{eqn: constraint violation maxLP}), we have
\begin{equation}
\Prob[\tildeZ \ge \starz + ( \lambda + \delta ) \sqrts]
        \;\le\; 2ne^{-\lambda ^2 /4} + 2e^{-\delta ^2 /4}.
    \label{eqn: RDM tildeZ bound}
\end{equation}

Choose $\lambda=\sqrt{4 \ln (8n)}$ and $\delta=\sqrt{4 \ln 8}$.
(Note that $\delta \le \sqrts$, since $s\ge 9$.)
When $\lambda \le \sqrt{s/2 - \starz}$, from (\ref{eqn: RDM tildeZ bound}),
with probability at least $\half$, we get
\begin{equation}
\tildeZ \le \starz + ( \sqrt{4 \ln (8n)} + \sqrt{4 \ln 8} ) \sqrts.
    \label{eqn: RDM final inequ}
\end{equation}

If $s< 4\sqrt{\ln (8n)}$, then $\sqrt{4s \ln (8n)} > s/2$,
inequality (\ref{eqn: RDM final inequ}) will be trivially true. Suppose
$s\ge 4\sqrt{\ln (8n)}$ and $\starz > s/2 - \sqrt{4\ln (8n)}$
(i.e., $\lambda > \sqrt{s/2 - \starz}$).
Then $\starz + \sqrt{4s \ln (8n)} \ge s/2$,
and inequality (\ref{eqn: RDM final inequ}) is also trivially true.
Thus by (\ref{eqn: RDM final inequ}) together with the bound
$\Dmax{\ast}(G,s) \ge \starz$, we obtain the lemma.
\end{proof}

%%%%%%%%%%%%%%%%%%%%%%%%%%%%%%%%%%%%%%%%%

\subsection{Algorithm~{\RCA} for {\BCPCavg}}
\label{sec: RCA algorithm}
In this section we present our randomized algorithm {\RCA}
for {\BCPCavg}.
Given $G = (C,P,E)$, again let $A$ be the Boolean $m\times n$
adjacency matrix of $G$.
Then {\BCPCavg} is equivalent to the following
integer linear program {\avgIP}:
\begin{eqnarray*}
\mbox{\rm maximize:}
    \quad \frac{1}{n}\sum_{j=1}^n z_j &&
                        \\
\mbox{\rm subject to:}
    \quad z_j \;&\le\;& \textstyle{\sum_{i=1}^m a_{ij} x_i} \quad\forall j=1,...,n
                        \\
    \quad z_j \;&\le\;& \textstyle{\sum_{i=1}^m (1-a_{ij}) x_i} \quad\forall j=1,...,n
                        \\
    \quad \textstyle{\sum_{i=1}^m x_i} \;&\le\;& s
                        \\
        x_i \;&\in\;& \braced{0,1} \quad\forall i=1,...,m
\end{eqnarray*}
The Boolean variables $x_i$ indicate whether the
corresponding $c_i \in C$ are selected or not.

\myparagraph{Algorithm~{\RCA}.} The algorithm first relaxes the
last constraint to $0\leq x_i \leq 1$ to obtain the linear program
{\avgLP}, and then computes an optimal solution $\starx_i$,
$i=1,2,...,m$, of {\avgLP}. Next, applying randomized rounding,
{\RCA} computes an integral solution $X_1,...,X_m$ by choosing
$X_i = 1$ with probability $\starx_i$ and $0$ otherwise. Note that
this solution may not be feasible since $\sum_{i=1}^m X_i$ may
exceed $s$. Let $L = \max \{\sum_{i=1}^m X_i - s, 0\}$. {\RCA}
changes $L$ arbitrary variables $X_i = 1$ to $0$, obtaining a
feasible solution $\tildeX_1,...,\tildeX_m$.

%%%%%%%%%%%%%%%%%%%%%

One can show that, in expectation,
for any instance $\angled{G,s}$,
{\RCA} finds a solution with objective value at least
$\Cavg{\ast}(G,s) - O( \sqrts )$.
We omit the proof because in the next section
we provide an algorithm with a better asymptotic bound.

%%%%%%%%%%%%%%%%%%%%%%%%%%%%%%%%%%%%%%%%%%%%%%%%%%%%%%%%
%%%%%%%%%%%%%%%%%%%%%%%%%%%%%%%%%%%%%%%%%%%%%%%%%%%%%%%%
%%%%%%%%%%%%%%%%%%%%%%%%%%%%%%%%%%%%%%%%%%%%%%%%%%%%%%%%

\subsection{An Alternative Algorithm {\RCAtwo} for {\BCPCavg}}
\label{sec: RCAtwo algorithm}

We now modify Algorithm~{\RCA}, to improve its approximation
bound.
Let $\starz$ be the optimum solution of {\avgLP},
that is
$\starz = \frac{1}{n}\sum_{j=1}^n
        \min \{\sum_{i=1}^m a_{ij}\starx_i, \sum_{i=1}^m (1- a_{ij}) \starx_i\}$.
Our new Algorithm~{\RCAtwo} is identical to {\RCA} in all steps except for
the rounding scheme:
choose $X_i = 1$ with probability $\frac{\starx_i}{1+\lambda}$
and $0$ otherwise, where $\lambda = \frac{1}{\sqrtstarz}$
(without loss of generality, assuming $\starz > 0$).

Before we start {\RCAtwo}'s analysis,
we state and prove a variant of the Chernoff bound needed to
estimate the error introduced by changing $L$ variables $X_i$
at the end of the algorithm.

%%%%%%%%%%%%%%%%%%%%

\begin{lemma}\label{lem: Chernoff-Young}
Let $Y_1, Y_2, ..., Y_n$ be $n$ independent Bernoulli trials,
where $\Pr [Y_i = 1]=p_i$. Then if $Y= \sum_{i=1}^n Y_i$ and if
$\Exp[Y] =\sum_i p_i \le \mu$, for any $0< \epsilon \le 1$:
\begin{equation}
\Exp[\max \{0, Y-(1+\epsilon)\mu\}]
\;\le\;
   \frac{2e^{-\mu \epsilon ^2 /4}} {\ln (1+\epsilon)}.
        \label{eqn: Chernoff-Young}
\end{equation}
\end{lemma}

\begin{proof}
See Appendix~\ref{sec: proof of chernoff-young}.
\end{proof}

\myparagraph{Analysis.} \label{sec: analysis RCA2}
 We denote by $\Cavg{\malyRCAtwo}(G,s)$ or $\tildeZ$ the value of the objective
function computed by {\RCAtwo} for {\BCPCavg}, that is $\Cavg{\malyRCAtwo}(G,s) =
    \tildeZ = \frac{1}{n}\sum_{j=1}^n \min
    \{\sum_{i=1}^m a_{ij} \tildeX_i,\sum_{i=1}^m (1- a_{ij}) \tildeX_i\}$.
Recall $\Cavg{\ast}(G,s)$ is the optimal value of $\Cavg{}(G,s)$
of {\BCPCavg}.

%%%%%%%%%%%%%%%%%%%%%
\begin{lemma}\label{lem: RCAtwo app bound}
For any instance $\angled{G,s}$ of {\BCPCavg},
\begin{equation}
\Exp[\Cavg{\malyRCAtwo}(G,s)]
\;\ge\;
    \Cavg{\ast}(G,s) -
    O\left( \sqrt{\Cavg{\ast}(G,s)} \right).
        \label{eqn: RCAtwo app bound}
\end{equation}
\end{lemma}

\begin{proof}
We can assume that $\Cavg{\ast}(G,s) \ge 1$,
because otherwise (\ref{eqn: RCAtwo app bound}) is trivially true.
Thus $\starz \ge 1$ as well, since $\starz \ge \Cavg{\ast}(G,s)$.

For all $j$ define constants
$\starz_j = \min \{\sum_{i=1}^m a_{ij}\starx_i, \sum_{i=1}^m (1- a_{ij}) \starx_i\}$
and variables
$\barZ_j = \sum_{i=1}^m a_{ij} X_i$,
$\hatZ_j = \sum_{i=1}^m (1- a_{ij}) X_i$ and
$Z_j = \min \{\barZ_j, \hatZ_j\}$.
Thus we have
$\starz = \frac{1}{n}\sum_{j=1}^n \starz_j$
and
$Z = \frac{1}{n}\sum_{j=1}^n Z_j$.

The $\{X_i\}$ are independent Bernoulli random variables
with $\Exp[X_i]=\frac{\starx_i}{1+\lambda}$.
So $\Exp[\barZ_j] \ge \frac{\starz_j}{1+\lambda}$ and
$\Exp[\hatZ_j] \ge \frac{\starz_j}{1+\lambda}$, for each $j$.
Applying the Chernoff-Wald bound \cite{young00kmedians}, we get
\begin{eqnarray*}
 &&\textstyle
 \Exp[(1-\epsilon)\frac{\starz_j}{1+\lambda} - (1+\epsilon) Z_j]
 \;\le\; \\
 &&\Exp\left [\max \left\{ (1-\epsilon)\frac{\starz_j}{1+\lambda} - (1+\epsilon)
 \barZ_j , (1-\epsilon)\frac{\starz_j}{1+\lambda} - (1+\epsilon) \hatZ_j \right\}\right ]
 \;\le\;
 \frac{\ln 2}{\epsilon},
\end{eqnarray*}
where $0< \epsilon \le \half$. Since
$\frac{1-\epsilon}{1+\epsilon} \ge 1-2\epsilon$,
$\Exp[Z_j] \;\ge\; \frac{\starz_j}{1+\lambda} -
 2 \left(\frac{\epsilon \starz_j}{1+\lambda} + \frac{1}{\epsilon} \right) $,
and thus we have
\begin{equation*}
\Exp[Z] \;\ge\; \frac{\starz}{1+\lambda} -
2\left(\frac{\epsilon \starz}{1+\lambda} + \frac{1}{\epsilon} \right).
\end{equation*}
In the above inequality we substitute $\lambda = \frac{1}{\sqrtstarz}$
and choose
$\epsilon = \frac{1}{2\sqrt {\starz}}$, which, by simple algebra,
yields
\begin{equation}
 \Exp[Z] \;\ge\; \starz - 6\sqrtstarz.
\label{eqn: bcpavgtwo Z bound}
\end{equation}
Likewise,
$\Exp[\sum_{i=1}^m X_i] = \frac{1}{1+\lambda}\sum_{i=1}^m \starx_i
\le \frac{s}{1+\lambda}$.
By Lemma~\ref{lem: Chernoff-Young}, we have
\begin{equation*}
\textstyle
 \Exp[\max\{0, \sum_{i=1}^m X_i - (1+\delta) \frac{s}{1+\lambda}\}]
 \;\le\; \frac{2 e^{- \frac{s}{1+\lambda} \delta ^2 /4} } {\ln (1+\delta)},
\end{equation*}
where $0< \delta \le 1$.
Letting $\delta = \frac{1}{\sqrtstarz}$ and substituting
$\lambda = \frac{1}{\sqrtstarz}$, the above inequality
implies
\begin{equation}
\Exp[L]
\;\le\;
    \frac{2 e^{- \frac{s}{4( \starz + \sqrtstarz)}}}
        { \ln (1+\frac{1}{\sqrtstarz})}
\;\le\;
    \frac{2}{ \ln (1+\frac{1}{\sqrtstarz})}
\;\le\;
    4\sqrtstarz,
\label{eqn: RCAtwo L bound}
\end{equation}
where the last inequality follows from $\ln(1+\epsilon) \ge \epsilon/2$
for $0 < \epsilon \le 1$.
Combining (\ref{eqn: bcpavgtwo Z bound}), (\ref{eqn: RCAtwo L bound}),
 and $\tildeZ \ge Z-L$, we get
\begin{equation*}
\Exp[\tildeZ] \;\ge\; \Exp[Z-L]
            \;\ge\; \starz - 10 \sqrtstarz.
\end{equation*}
Since $1\le \Cavg{\ast}(G,s)\le \starz$, the above bound implies
(\ref{eqn: RCAtwo app bound}).
\end{proof}

Note that performance bounds for {\RCAtwo} and {\RCA}
are weaker than those for the algorithms in the previous sections,
as it holds only in expectation.
Algorithm {\RCA}'s approximation error is
slightly worse than that of {\RCAtwo}. Nevertheless, our experimental
analysis (not included) show that on synthetic and real data sets,
{\RCAtwo} does not outperform {\RCA}.

%%%%%%%%%%%%%%%%%%%%%%%%%%%%%%%%%%%%%%%%%%%%%%%%%%%%%%%%
%%%%%%%%%%%%%%%%%%%%%%%%%%%%%%%%%%%%%%%%%%%%%%%%%%%%%%%%
%%%%%%%%%%%%%%%%%%%%%%%%%%%%%%%%%%%%%%%%%%%%%%%%%%%%%%%%

\section{Experimental Analysis}
\label{sec: experiments}

We implemented Algorithms~{\RCM}, {\RCMtwo} and {\RDM}
using \emph{LP\_SOLVE} solver
\cite{lpsolve}, and tested their
performance on both synthetic and real data.

%%%%%%%%%%%%%%%%%%%%

\myparagraph{Synthetic data.} We tested these three algorithms on
random data sets represented by adjacency matrices of four sizes
$(m,n) =$ $(100,30)$, $(100,100)$, $(200,60)$, $(200,200)$, where each
element of the matrix is chosen to be $1$ or $0$ with probability
$\half$. We ran these programs for $s=20, 21, ..., 90$, and compared
the solution to the optimal solution of the linear relaxation (that
is, {\RCM} and {\RCMtwo} were compared to {\minLP}, {\RDM} was
compared to {\maxLP}).

\begin{table}
%% increase table row spacing, adjust to taste
%\renewcommand{\arraystretch}{1.3}
\caption{Performance of {\RCM} on synthetic data with
         $m=100$ and $n=30$}
\label{tab2}
\begin{center}
\setlength{\tabcolsep}{3.7pt}
\setlength{\columnsep}{0mm}
{\footnotesize
\begin{tabular}{@{}cccccccccccccccc@{}}
    \hline
    {} &{} &{} &{} &{} {} &{} &{} &{} &{} {} &{} &{} &{} &{} &{}\\[-1.5ex]
    $s$ &20 &25 &30 &35 &40 &45 &50 &55 &60 &65 &70 &75 &80 &85 &90\\[1ex]
    \hline
    {} &{} &{} &{} &{} {} &{} &{} &{} &{} {} &{} &{} &{} &{} &{}\\[-1.5ex]
    {\minLP}  &10 &12.5 &15 &17.5 &20 &22.5 &25 &27.5 &29.82 &31.89 &33.88 &35.76 &37.54 &39.11 &40\\[1ex]
    {\RCM} &7 &10 &13 &15 &18 &19 &23 &25 &28 &30 &32 &34 &35 &38 &40\\[1ex]
    \hline
\end{tabular}
}
\end{center}
\end{table}

Table \ref{tab2} shows results of the comparison of {\RCM}'s
solution and the {\minLP} solution of {\BCPCmin} from the experiment
in which $m=100$ and $n=30$. This table presents only the
performance of a single run of {\RCM}, so the results are likely to
be even better if we run {\RCM} several times and choose the best
solution.

\begin{figure}[ht]
    \centering
    \begin{tabular}{cc}
    \begin{minipage}{2.4in}
    \includegraphics[height=1.5in, width=2.4in]{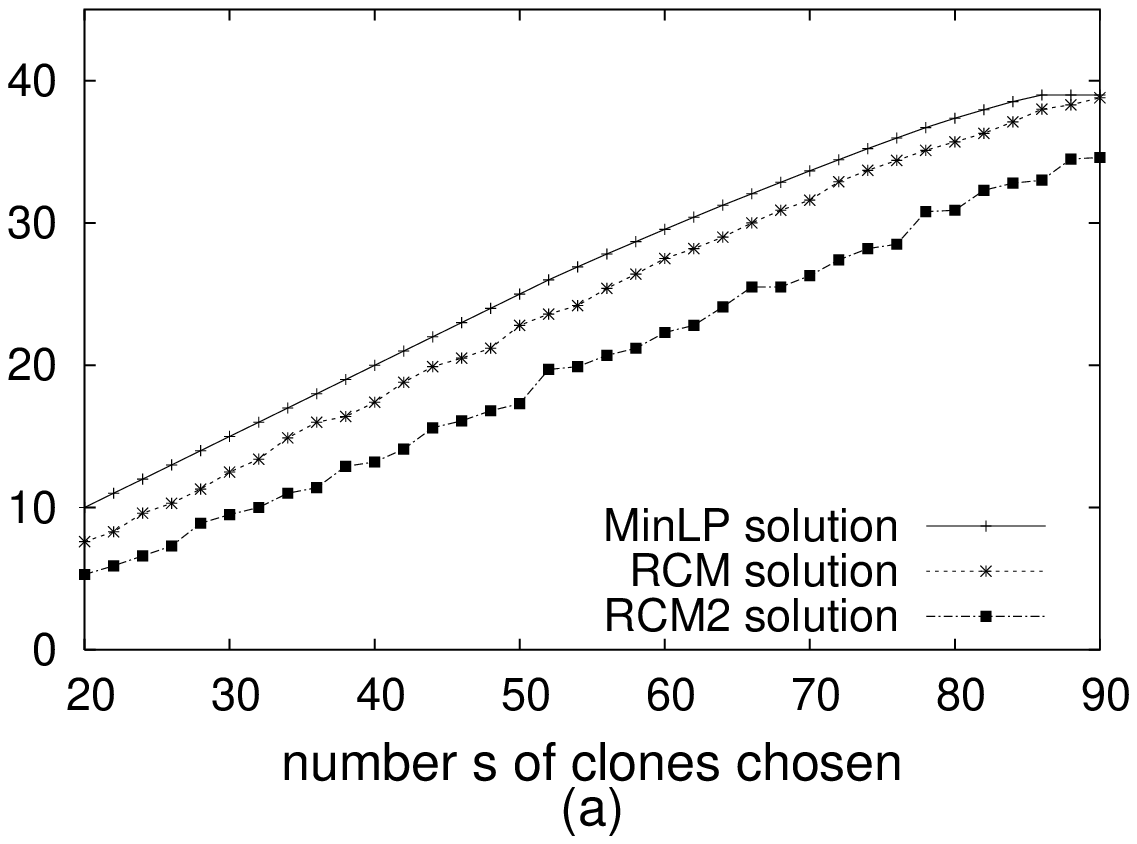}
    \end{minipage}
    &
    \begin{minipage}{2.4in}
    \includegraphics[height=1.5in, width=2.4in]{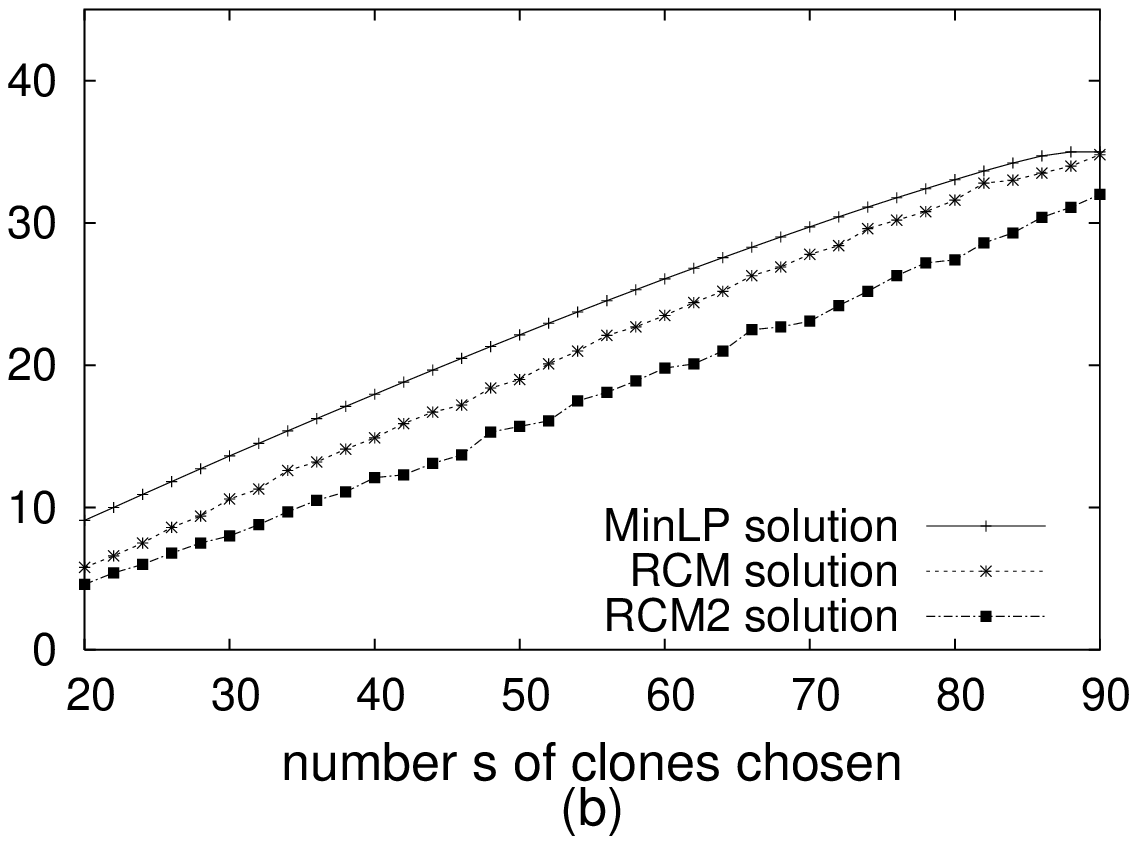}
    \end{minipage}
    \end{tabular}
\\
    \centering
    \begin{tabular}{cc}
    \begin{minipage}{2.4in}
    \includegraphics[height=1.5in, width=2.4in]{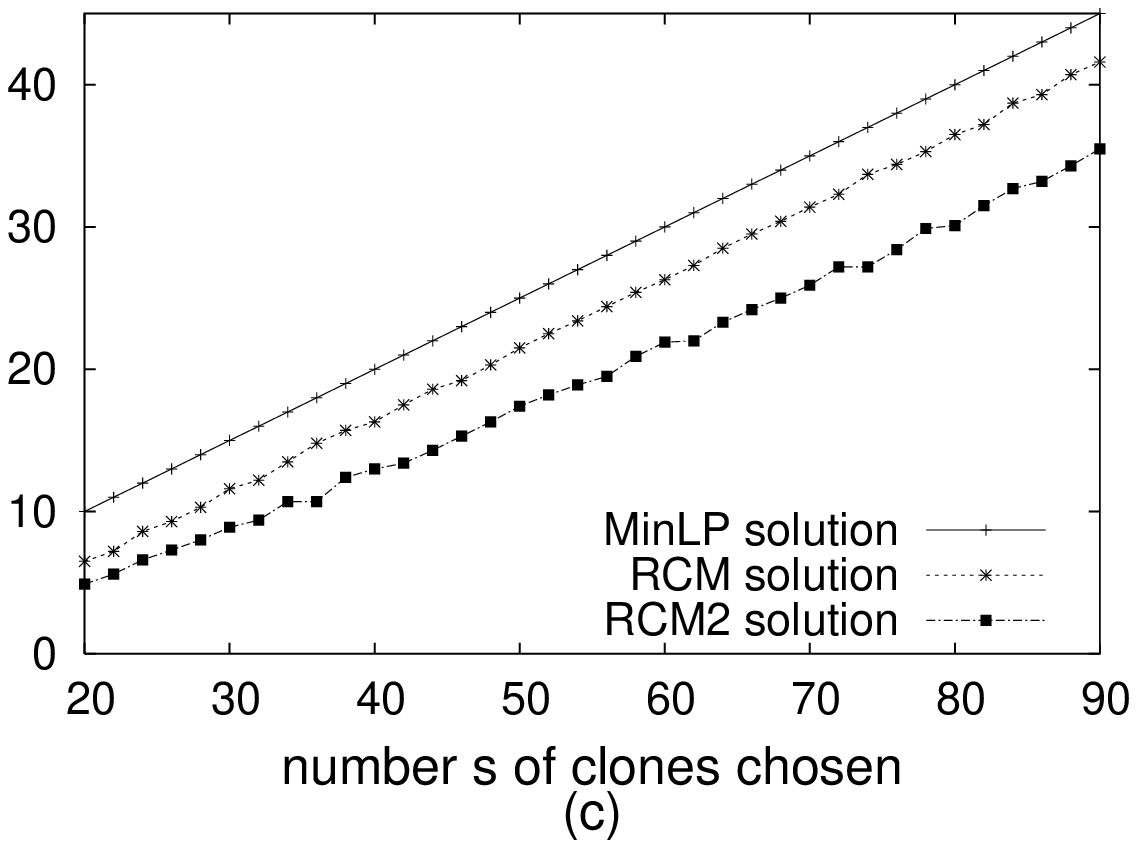}
    \end{minipage}
    &
    \begin{minipage}{2.4in}
    \includegraphics[height=1.5in, width=2.4in]{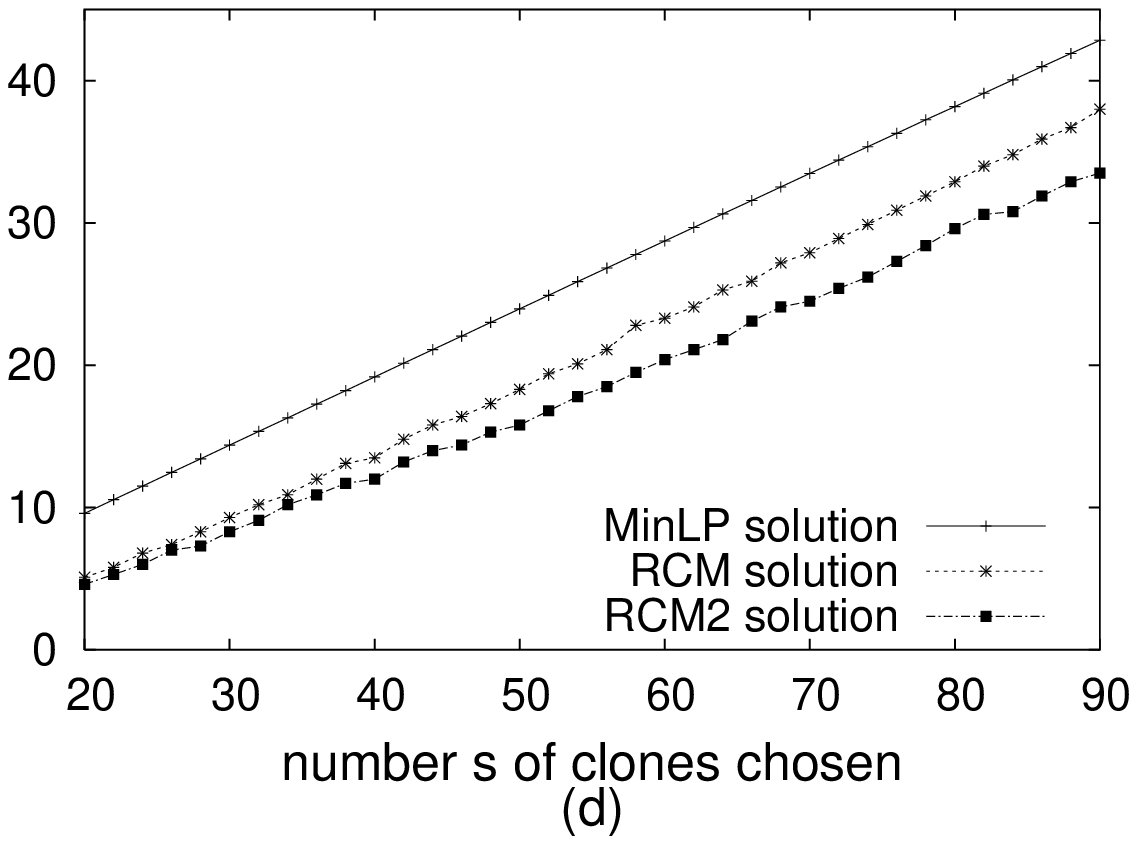}
    \end{minipage}
    \end{tabular}
    \caption{
    \small{
    {\RCM}'s and {\RCMtwo}'s performance for {\BCPCmin} on synthetic data
        for four matrices:
    (a) $(m,n) = (100,30)$;
        (b) $(m,n) = (100,100)$;
    (c) $(m,n) = (200,60)$;
        (d) $(m,n) = (200,200)$.
    The y-axis in the graph represents the objective value.
    }
    }
    \label{fig1}
\end{figure}

\begin{figure}[ht]
    \centering
    \begin{tabular}{cc}
    \begin{minipage}{2.4in}
    \includegraphics[height=1.5in, width=2.4in]{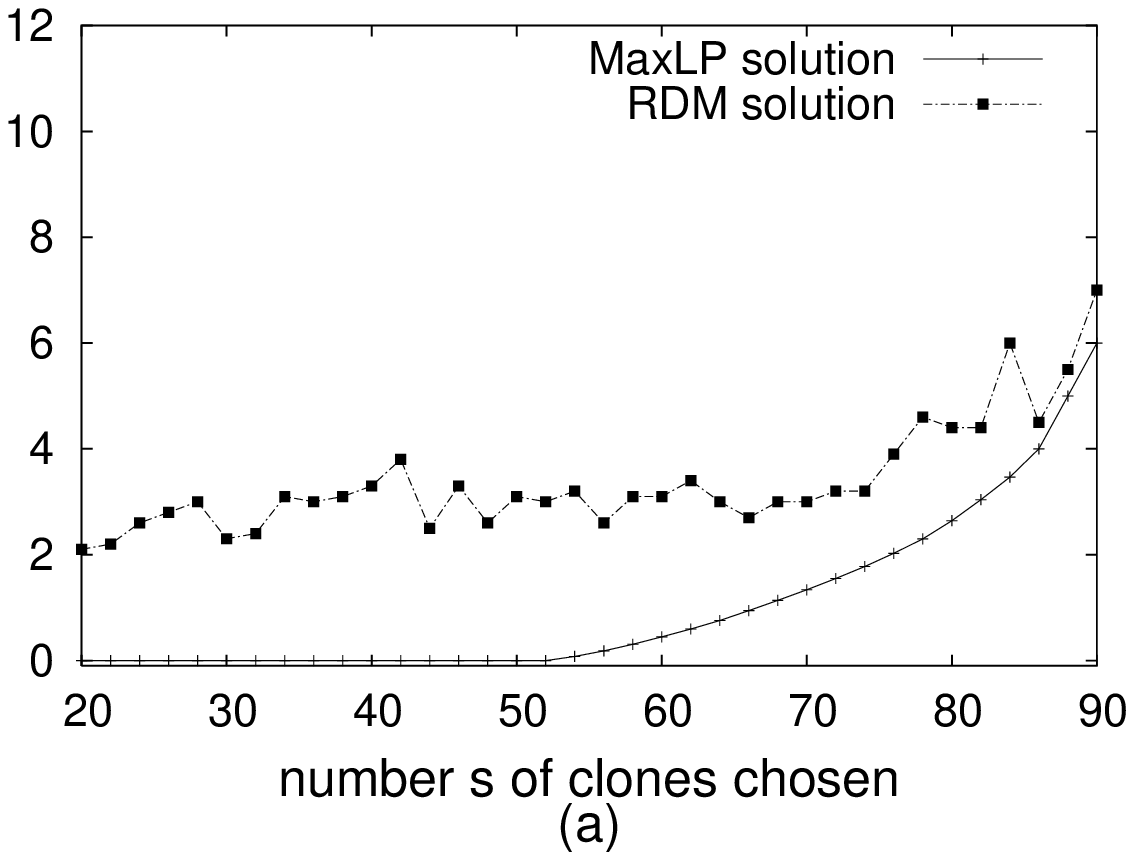}
    \end{minipage}
    &
    \begin{minipage}{2.4in}
    \includegraphics[height=1.5in, width=2.4in]{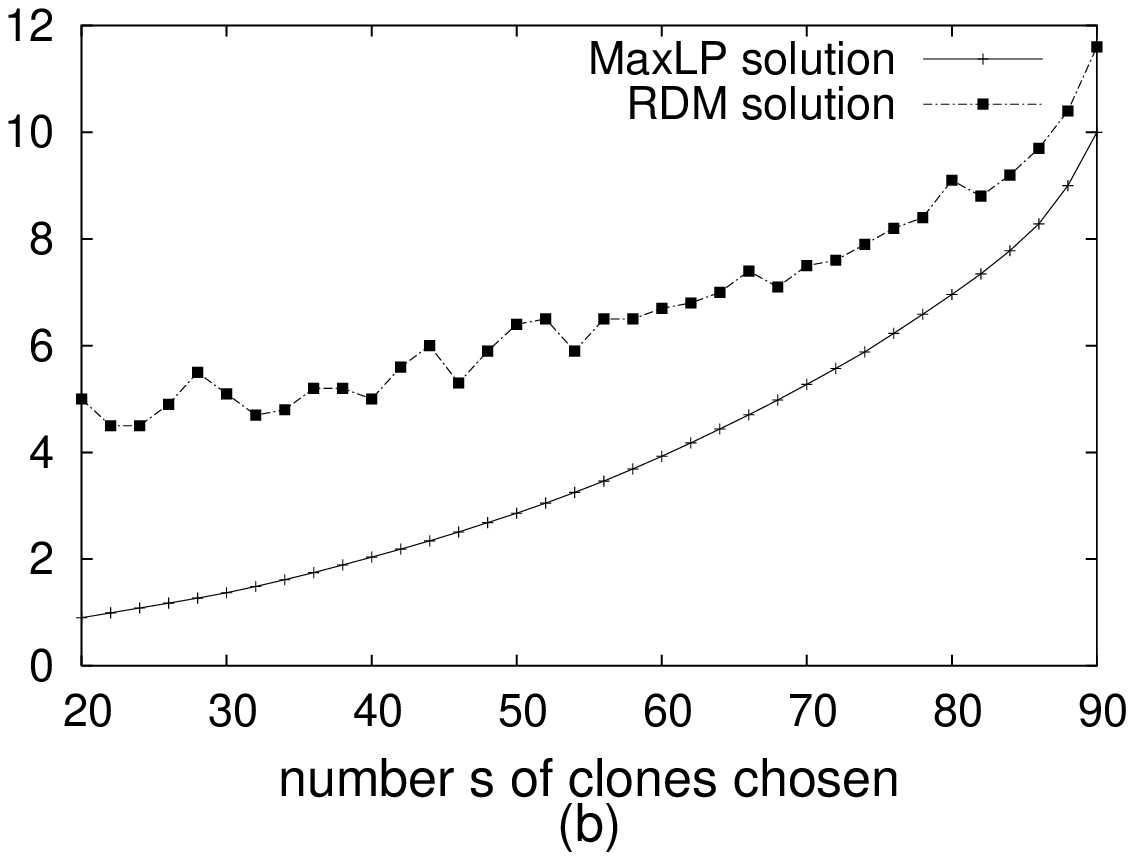}
    \end{minipage}
    \end{tabular}
\\
    \centering
    \begin{tabular}{cc}
    \begin{minipage}{2.4in}
    \includegraphics[height=1.5in, width=2.4in]{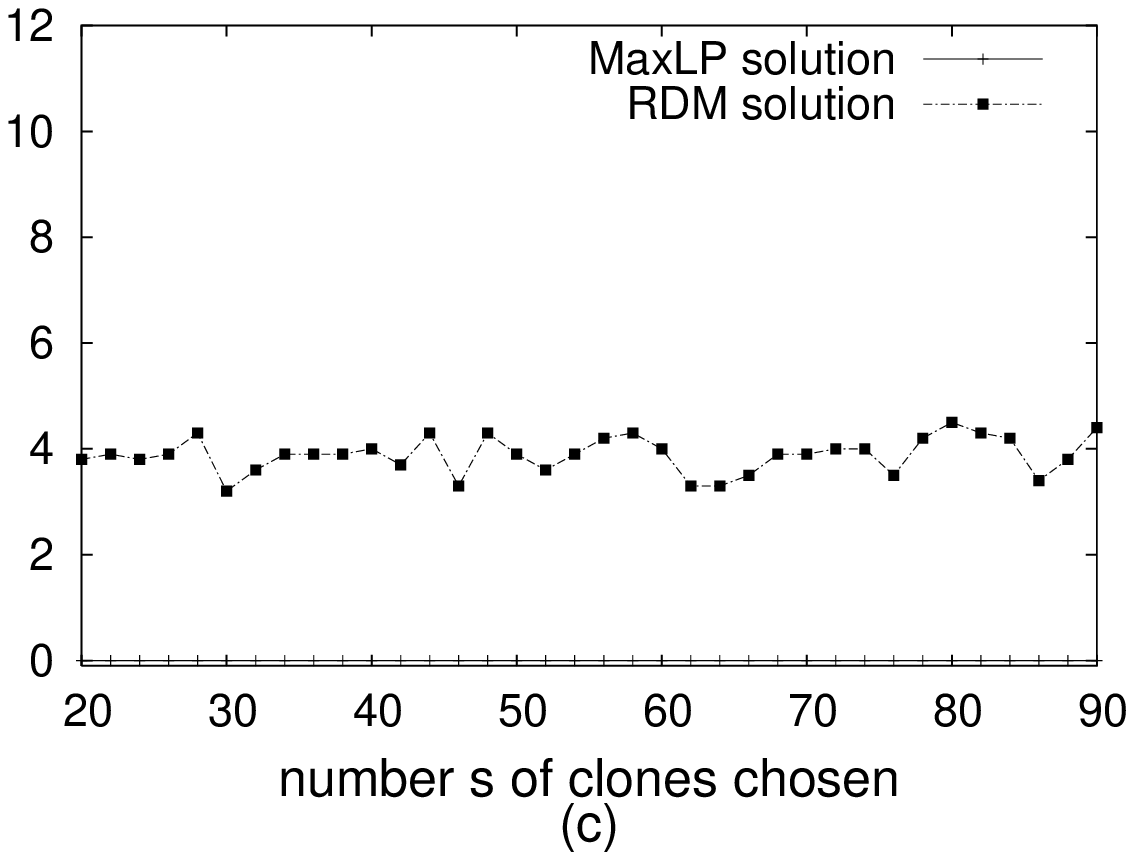}
    \end{minipage}
    &
    \begin{minipage}{2.4in}
    \includegraphics[height=1.5in, width=2.4in]{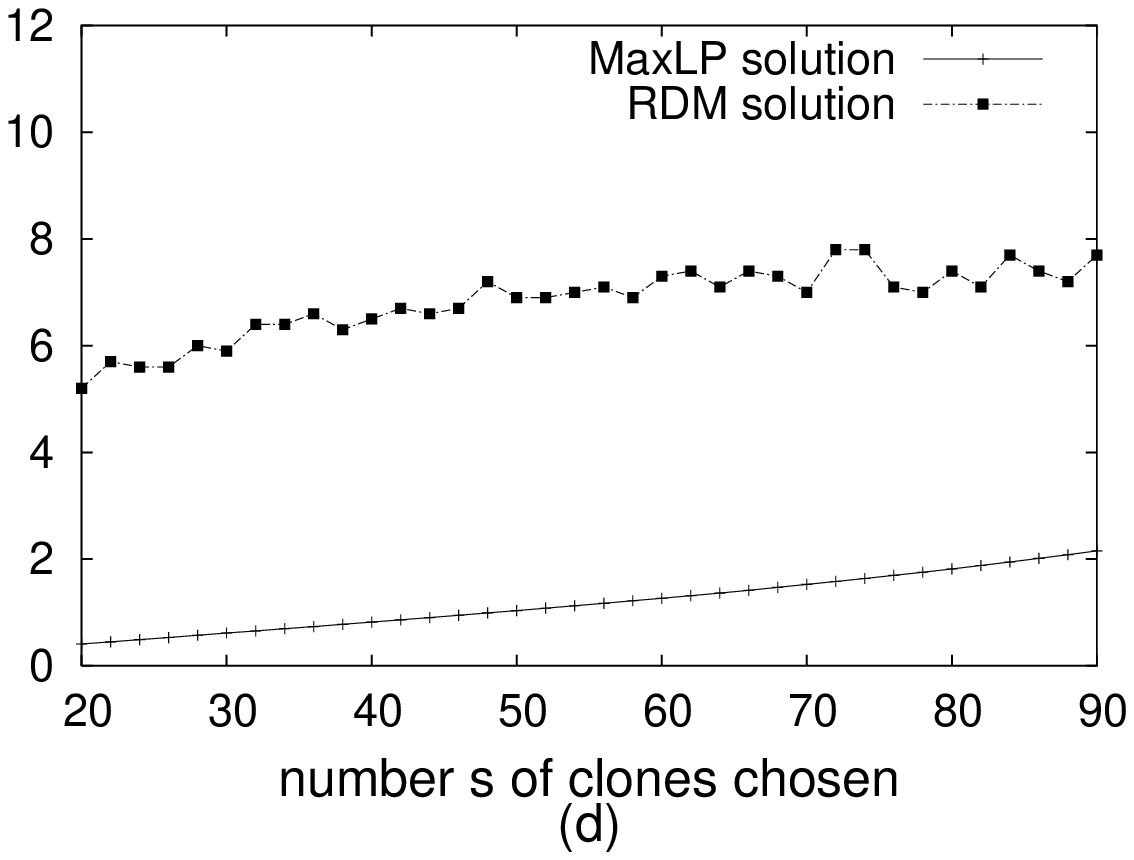}
    \end{minipage}
    \end{tabular}
    \caption{
    \small{
    {\RDM}'s performance for {\BCPDmax} on synthetic data
        for four matrices:
    (a) $(m,n) = (100,30)$;
        (b) $(m,n) = (100,100)$;
    (c) $(m,n) = (200,60)$;
        (d) $(m,n) = (200,200)$.
    The y-axis in the graph represents the objective value.
    }
    }
    \label{fig2}
\end{figure}

We also repeated our simulation test $10$ times
for each of the above settings and took the average of them.
Figure~\ref{fig1} and Figure~\ref{fig2} illustrate results of
these experiments.

It is worth observing that (on average)
{\RCM} was always able to find the solution close to that of
{\minLP}. Furthermore, since the true optimum (integral solution)
for {\BCPCmin} could be smaller than the solution for {\minLP}, our
approximation of {\RCM} could be even closer to the true optimum than
it appears. These observations apply to {\RDM} as well.

We have repeated these experiments for sparser random matrices, where
the values in the matrix were chosen to be $1$ with 
probability $\onefourth$ or $\oneeighth$. In all these experiments
the results were very similar to those for the distribution with
probability $\half$.

%%%%%%%%%%%%

\myparagraph{Real data.} To test the performance of these three
algorithms on real data, we used four clone-probe adjacency matrices.
The first two matrices represent hybridization of $500$ bacterial
clones extracted from rRNA genes analyzed in \cite{Valin02b} with
two sets of $30$ and $40$ probes designed with the algorithm
in \cite{James01}. The other two matrices (with similar parameters)
represent hybridization of rRNA genes from fungal clones analyzed
in \cite{Valin02a} with their corresponding sets of probes. For each
of these four data sets (and for each $s=200, 210, ..., 400$)
we tested {\RCM} $10$ times, and took the average of them. We observe
that {\RCM} found the solution with value at least $97\%$ of
{\minLP}'s solution in $92.8\%$ cases. We also repeated our tests for
{\RCMtwo} and {\RDM} using the above data sets with the same
settings. The results are summarized in Figure \ref{fig3} and Figure
\ref{fig4}.

\begin{figure}[ht]
    \centering
    \begin{tabular}{cc}
    \begin{minipage}{2.4in}
    \includegraphics[height=1.5in, width=2.4in]{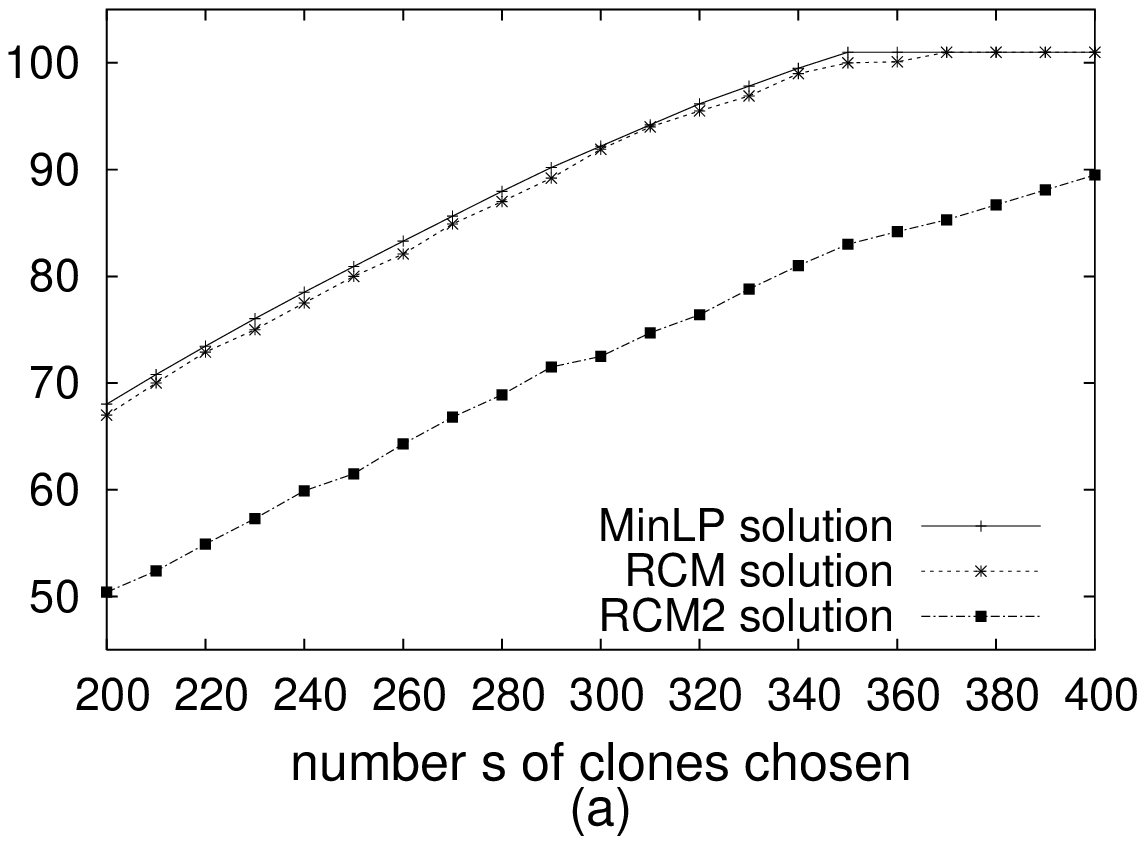}
    \end{minipage}
    &
    \begin{minipage}{2.4in}
    \includegraphics[height=1.5in, width=2.4in]{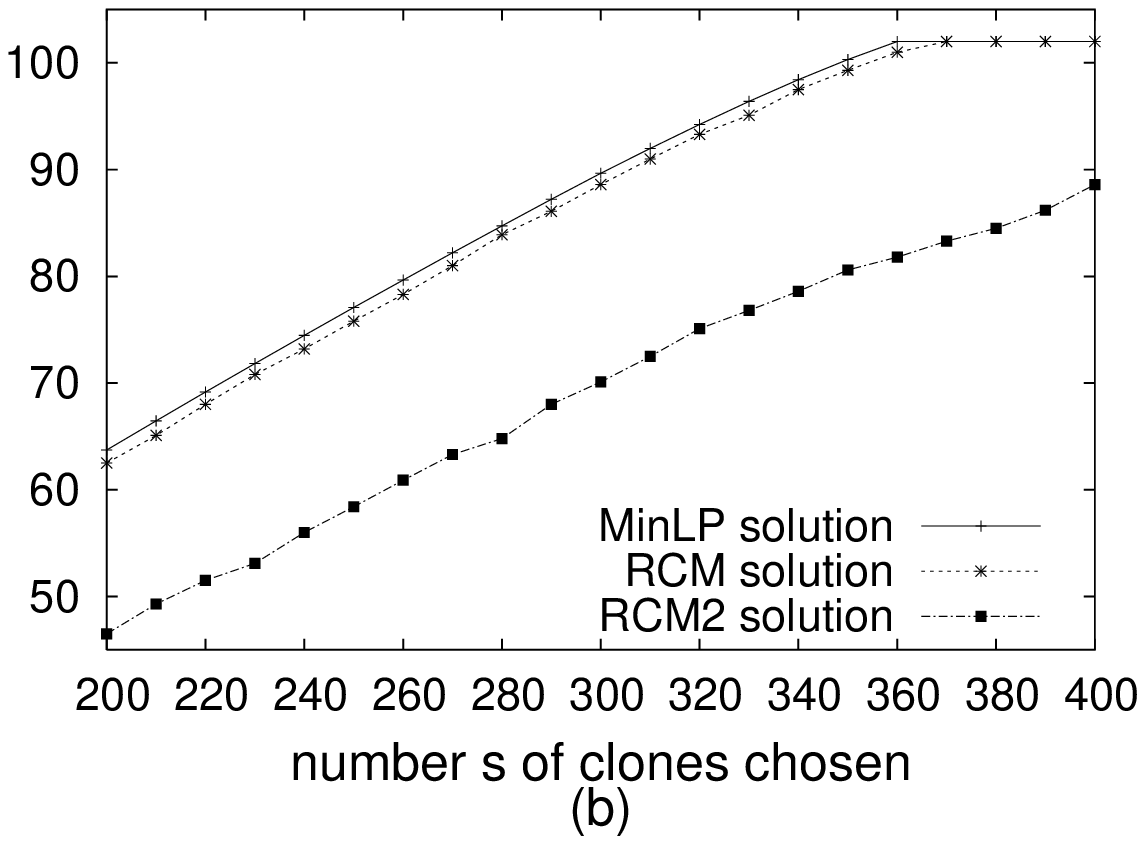}
    \end{minipage}
    \end{tabular}
\\
    \centering
    \begin{tabular}{cc}
    \begin{minipage}{2.4in}
    \includegraphics[height=1.5in, width=2.4in]{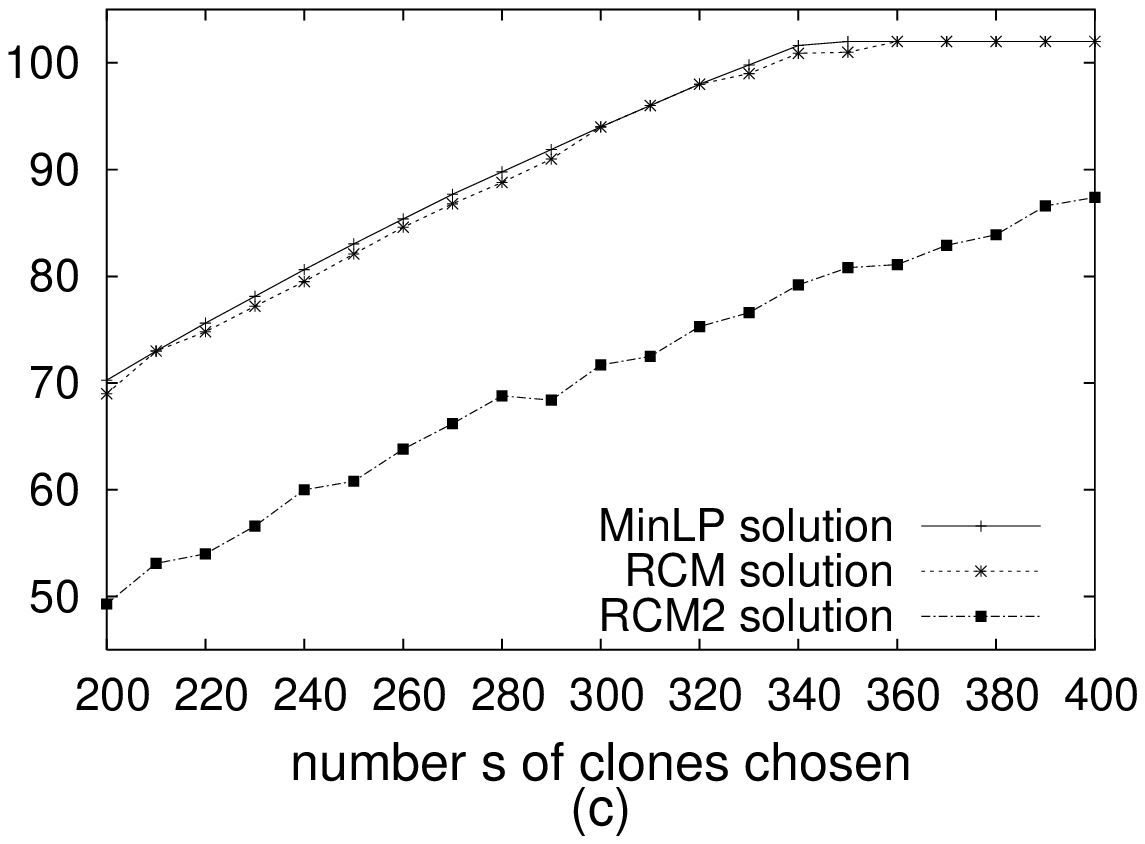}
    \end{minipage}
    &
    \begin{minipage}{2.4in}
    \includegraphics[height=1.5in, width=2.4in]{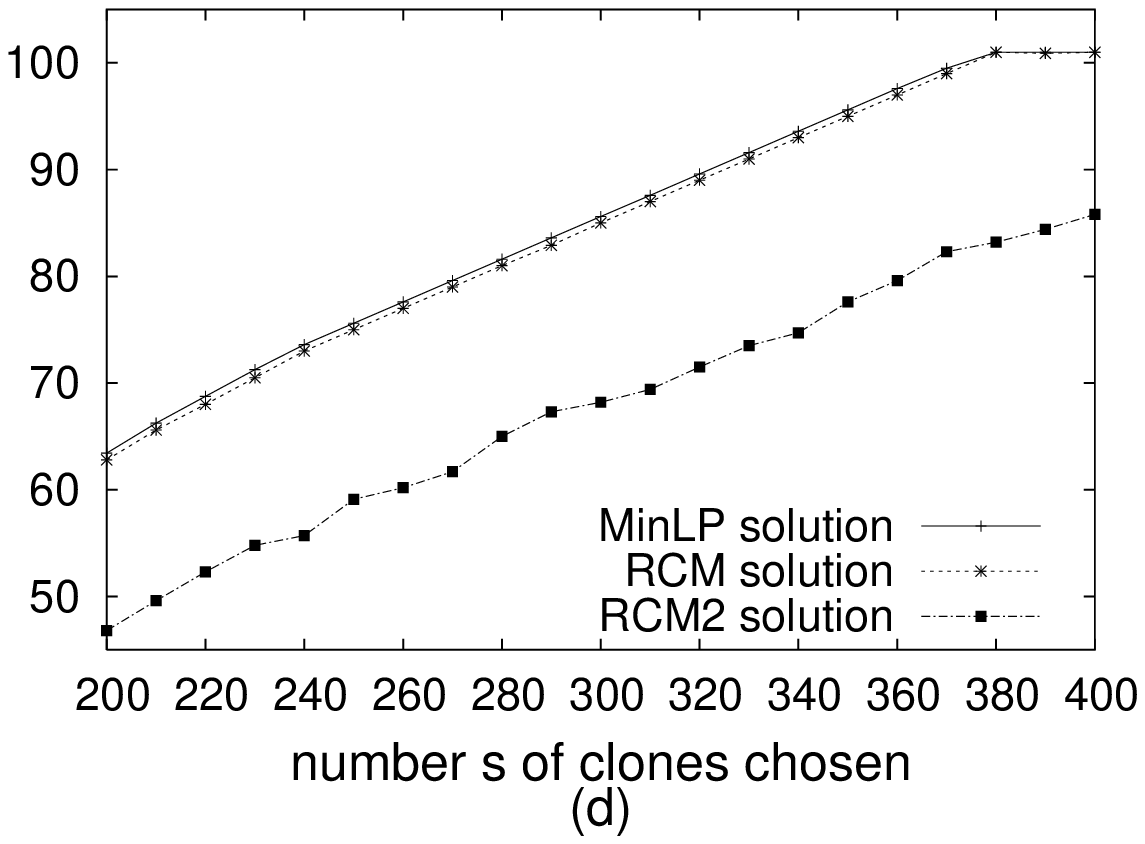}
    \end{minipage}
    \end{tabular}
    \caption{
    \small{
    {\RCM}'s and {\RCMtwo}'s performance on real data:
    (a) $500$ bacterial clones and $30$ probes;
        (b) $500$ bacterial clones and $40$ probes;
    (c) $500$ fungal clones and $30$ probes;
        (d) $500$ fungal clones and $40$ probes.
    The y-axis in the graph represents the objective value.
    }
    }
    \label{fig3}
\end{figure}

\begin{figure}[ht]
    \centering
    \begin{tabular}{cc}
    \begin{minipage}{2.4in}
    \includegraphics[height=1.5in, width=2.4in]{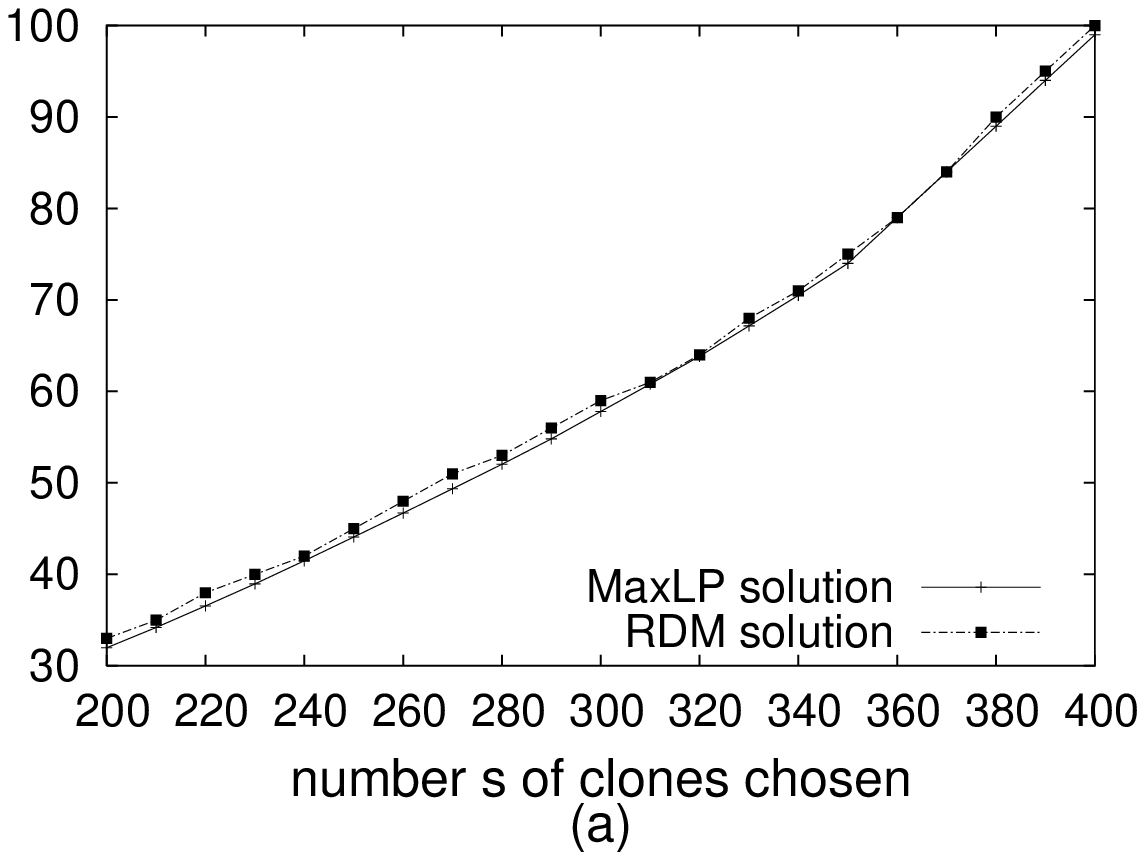}
    \end{minipage}
    &
    \begin{minipage}{2.4in}
    \includegraphics[height=1.5in, width=2.4in]{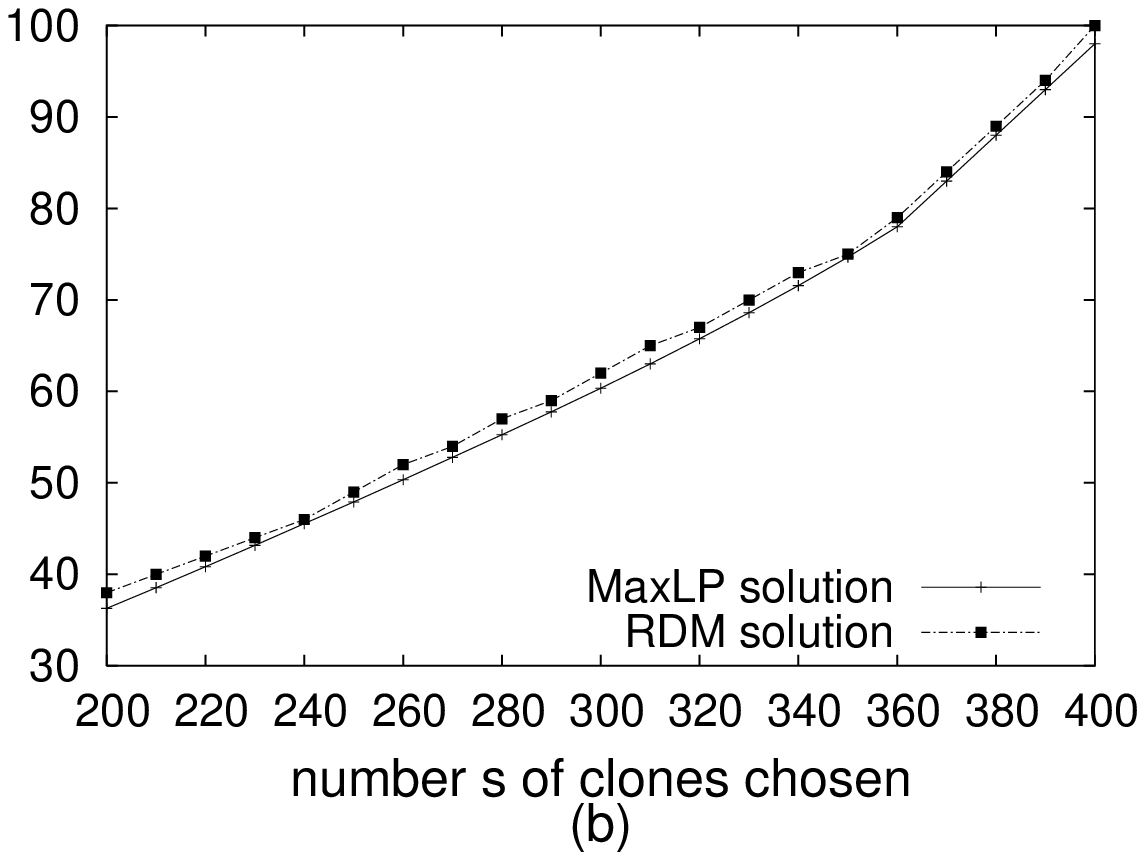}
    \end{minipage}
    \end{tabular}
\\
    \centering
    \begin{tabular}{cc}
    \begin{minipage}{2.4in}
    \includegraphics[height=1.5in, width=2.4in]{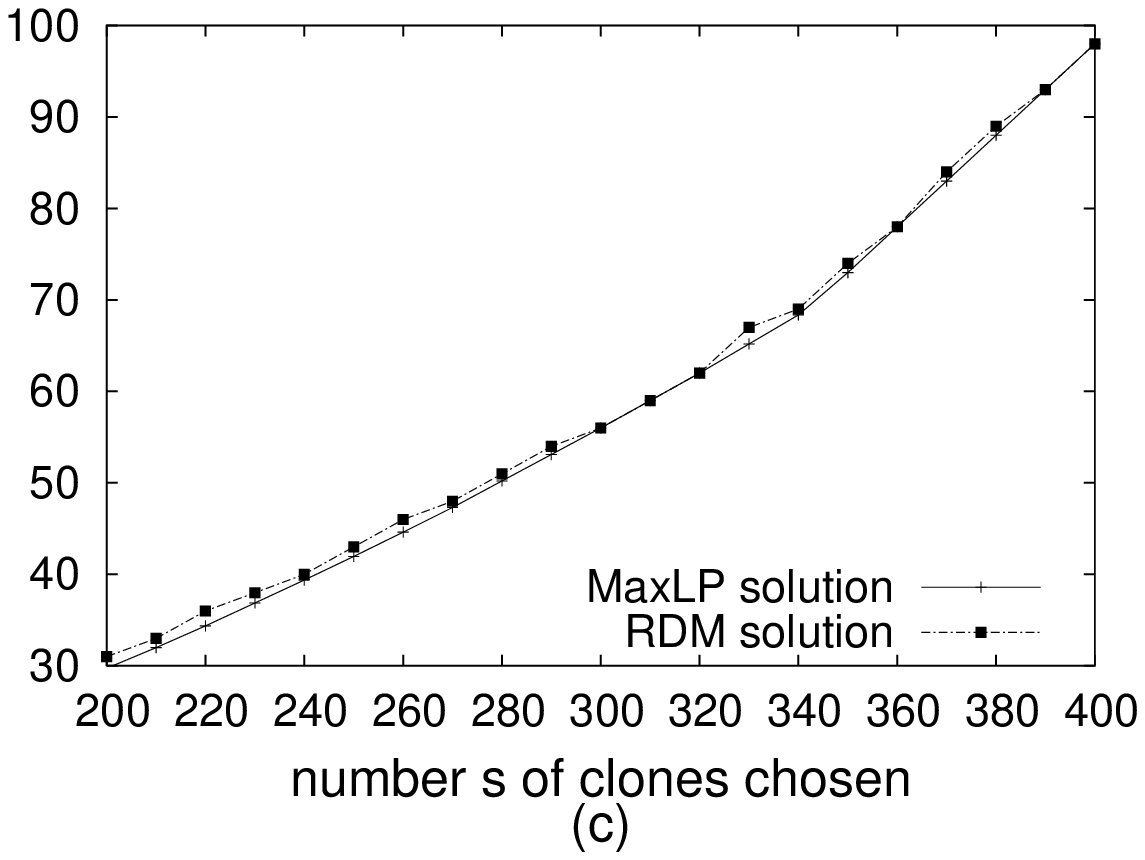}
    \end{minipage}
    &
    \begin{minipage}{2.4in}
    \includegraphics[height=1.5in, width=2.4in]{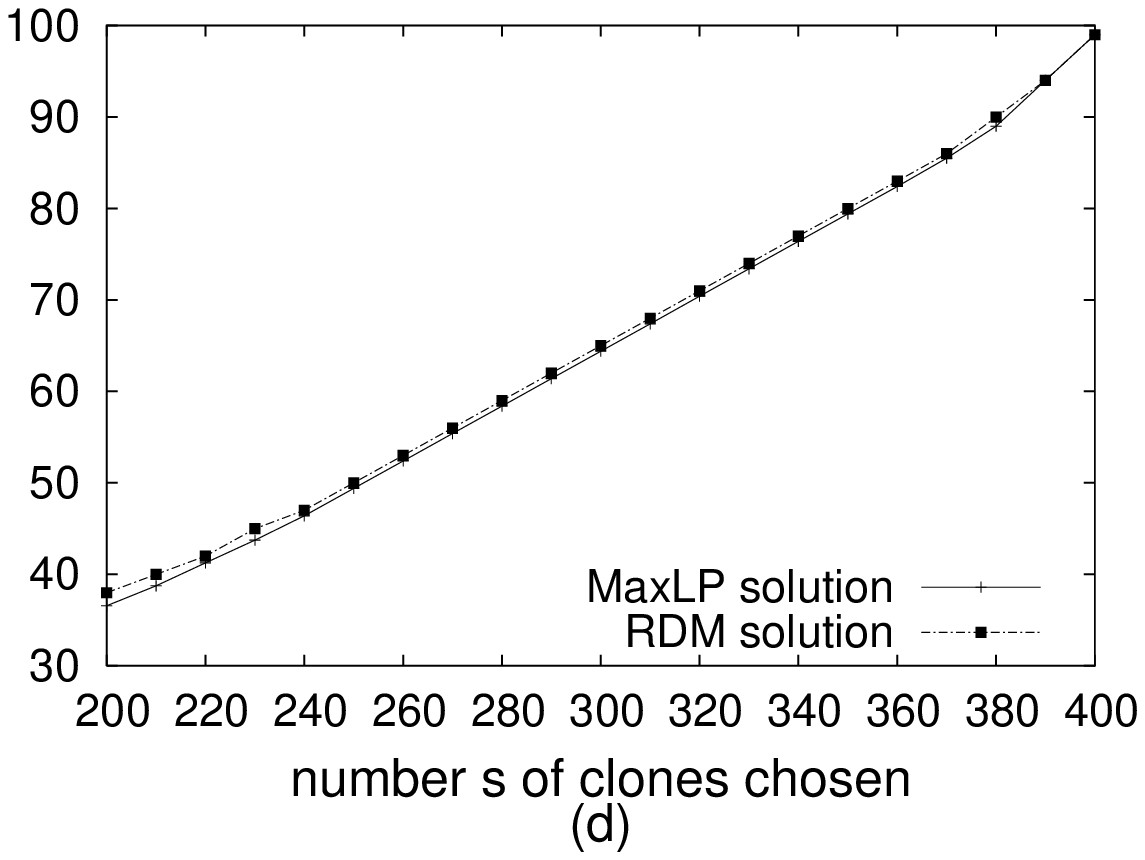}
    \end{minipage}
    \end{tabular}
    \caption{
    \small{
    {\RDM}'s performance on real data:
    (a) $500$ bacterial clones and $30$ probes;
        (b) $500$ bacterial clones and $40$ probes;
    (c) $500$ fungal clones and $30$ probes;
        (d) $500$ fungal clones and $40$ probes.
    The y-axis in the graph represents the objective value.
    }
    }
    \label{fig4}
\end{figure}

Figure~\ref{fig3} and Figure~\ref{fig4} suggest that solutions
found by {\RCM} and {\RDM} on real data are even closer to {\minLP} and
{\maxLP} solutions, respectively, than those for synthetic data,
sometimes even coinciding with {\minLP} and {\maxLP} solutions; i.e.,
{\RCM} and {\RDM} achieved optimum in some cases.

We also performed experimental analysis for {\RCA} and {\RCAtwo}
using the same synthetic and real data sets. The study
shows that {\RCA} and {\RCAtwo} approximate well
the optimum solutions. (The results are similar to those of
{\RCM} and {\RCMtwo} and are omitted.)

Our experimental results indicate that {\RCM} performs better than
{\RCMtwo} in practice even though, according to our
analysis in Section~\ref{sec: RCM algorithm} and~\ref{sec: RCMtwo algorithm},
{\RCMtwo} has a better asymptotic bound. This is likely to be
caused by a combination of several factors.
First, the constants in the asymptotic bounds for {\RCMtwo} appear
to be larger than those for {\RCM}, and our data sets may not be
large enough for the asymptotic trends to show.
Second, the bound for {\RCMtwo} is better than that for
{\RCM} only if the optimum is sufficiently small compared to $s/\log n$.
As the parameters of this range depend on the hidden asymptotic constants,
it is not clear whether our data sets are within this range.
Finally, if the number of clones initially selected by the algorithm
is less than $s$, our implementation of both algorithms
adds some arbitrary clones to increase their number to $s$.
Since {\RCMtwo} uses slightly smaller
probabilities in the rounding scheme, it tends to choose initially fewer
clones, and thus it is also likely to add more of these arbitrary
clones.
The performance of {\RCMtwo} relative to {\RCM} would probably
be improved with additional random sampling.
(The same arguments apply to {\RCA} and {\RCAtwo} as well.)

Our experiments were performed on a machine with Intel Pentium 4
$2.4$GHz CPU and $1$GB RAM. The total running time for each single
run of {\RCM} or {\RDM} on these synthetic and real data sets was
in the range of $20$-$80$ seconds, which is practically acceptable.

%%%%%%%%%%%%%%%%

\smallskip

\paragraph{Example.}
To complement the above statistics with a more concrete
example, we now describe the results
of {\RCM} on a typical data set. Here we used {\RCM} to
compute a set $D$ of $s = 100$ control clones for
$m = 500$ bacterial clones with a set of $n = 30$ probes.
The distribution of the degrees of the probes with
respect to $D$ is given in the table below:

\begin{center}
\begin{tabular}{l|c|c|c|c|c|c|}
degree &  0 --30 & 31 -- 40 & 41 -- 50 & 51 -- 60
                         & 61 -- 70 & 71 -- 100
                                 \\
n. of probes
                  & 0 & 14 & 7 & 4 & 5 & 0
\end{tabular}
\end{center}

The minimum and maximum degrees of probes in $D$ were
$37$ and $68$, respectively, thus producing the objective
value $\Cmin{}(D) = 32$ for this instance. Thus this $D$ is
a high quality control clone set.

%%%%%%%%%%%%%%%%%%%%%%%%%%%%%%%%%%%%%%%%%%%%%%%%%%%%%%%%
%%%%%%%%%%%%%%%%%%%%%%%%%%%%%%%%%%%%%%%%%%%%%%%%%%%%%%%%
%%%%%%%%%%%%%%%%%%%%%%%%%%%%%%%%%%%%%%%%%%%%%%%%%%%%%%%%

\section{Concluding Remarks}
\label{sec: conclusion}

We performed similar experiments for other algorithms provided
in this paper, and the results were equally promising. Overall, our
work demonstrates that randomized rounding is a very effective
method for solving all versions of Balanced Covering, especially on
real data sets. In the actual implementation available at
{\tt http://algorithms.cs.ucr.edu/OFRG/}, the solution of
{\RCM} is fed as
an initial solution into a simulated-annealing algorithm.
We found out that the simulated annealing rarely produces
any improvement of this initial solution, which provides
further evidence for the effectiveness of randomized rounding in this case.
(In contrast, when we run simulated annealing from a random
initial solution, in a typical run, it takes approximately 10 minutes
to find a solution that is about 80\% as good as that
of {\RCM}.)

We remark that (by creating two copies of the matrix and inverting
the bits in the second copy) {\BCPCmin} can be reduced to a more
general problem where we want to cover all columns with the maximum
number of $1$'s. Our algorithms and their analyses apply to this
problem as well.

%%%%%%%%%%%%%%%%%%%%%%%%%%%%%%%%%%%%%%%%%%%%%%%%%%%%%%%%
%%%%%%%%%%%%%%%%%%%%%%%%%%%%%%%%%%%%%%%%%%%%%%%%%%%%%%%%
%%%%%%%%%%%%%%%%%%%%%%%%%%%%%%%%%%%%%%%%%%%%%%%%%%%%%%%%

\section*{Acknowledgments}
This work is supported by NSF Grant BD\&I-0133265.
Work of Qi Fu and Marek Chrobak is partially
supported by NSF Grant CCR-0208856.

We would like to thank the anonymous referees for invaluable
suggestions that helped us improve the presentation of this work.

%%%%%%%%%%%%%%%%%%%%%%%%%%%%%%%%%%%%%%%%%%%%%%%%%%%%%%%%
%%%%%%%%%%%%%%%%%%%%%%%%%%%%%%%%%%%%%%%%%%%%%%%%%%%%%%%%
%%%%%%%%%%%%%%%%%%%%%%%%%%%%%%%%%%%%%%%%%%%%%%%%%%%%%%%%

\appendix
\section{Appendix: Proof of Lemma~\ref{lem: Chernoff-Young}}
\label{sec: proof of chernoff-young}

\begin{proof}
Let $c(\epsilon) = e^{\epsilon} / (1+\epsilon)^{1+\epsilon}$
and $f(x) = -\ln (c(x)) / x$. We will first show
\begin{equation}
 \int_{x=\epsilon}^{+\infty} c(x)^\mu \,dx
 \;\le\;
 \frac {2c(\epsilon)^\mu} {\mu \ln (1+\epsilon)}.
 \label{eqn: Chernoff-Young claim}
\end{equation}
Integrating both sides of the inequality $1+\ln (1+x) \le 1+x$, for $x\ge 0$,
we get $(1+x)\ln (1+x) \le x(1+x/2)$. By simple algebra, we then get
$f'(x) \ge (\ln (1+x)/2)'$, and thus $f(x)$
is an increasing function and $f(x) \ge \ln (1+x) /2$.
We can now verify (\ref{eqn: Chernoff-Young claim}) as follows,
\begin{eqnarray*}
 \int_{x=\epsilon}^{+\infty} c(x)^\mu \,dx
    & =   & \int_{x=\epsilon}^{+\infty} e^{-\mu x f(x)} \,dx \\
    & \le & \int_{x=\epsilon}^{+\infty} e^{-\mu x f(\epsilon)} \,dx \\
    & =   & \frac{ e^{-\mu \epsilon f(\epsilon)} } {\mu f(\epsilon)} \\
    & =   & \frac{ c(\epsilon) ^\mu } {\mu f(\epsilon)} \\
    & \le & \frac{ 2c(\epsilon) ^\mu } {\mu \ln (1+ \epsilon)}.
\end{eqnarray*}
We have
\begin{eqnarray*}
    \Exp[\max \{0, Y-(1+\epsilon)\mu \} ]
        &=& \int_{y=0}^{+\infty} \Pr[Y- (1+\epsilon)\mu \ge y] \,dy.
\end{eqnarray*}
Choose $\delta$ so that $(1+\delta)\mu = (1+\epsilon) \mu + y$. Changing variables
from $y$ to $(\delta-\epsilon)\mu$, and applying a standard Chernoff bound,
the expected value above becomes
\begin{eqnarray*}
    \Exp[\max \{0, Y-(1+\epsilon)\mu \} ]
     &=& \mu \int_{\delta = \epsilon}^{+\infty} \Pr[Y \ge (1+\delta)\mu] \,d\delta
                \\
    &\le& \mu \int_{\delta = \epsilon}^{+\infty} c(\delta)^\mu \,d\delta.
\end{eqnarray*}
Combining this, inequality (\ref{eqn: Chernoff-Young claim}) and the fact that
$c(\epsilon) ^\mu \le e^{-\mu \epsilon ^2 /4}$ for $0< \epsilon \le 1$,
we get
\begin{equation*}
 \Exp[\max \{0, Y-(1+\epsilon)\mu \} ]
 \;\le\;
 \frac{ 2c(\epsilon) ^\mu } {\ln (1+ \epsilon)}
 \;\le\;
 \frac{2e^{-\mu \epsilon ^2 /4}}{\ln (1+ \epsilon)},
\end{equation*}
and the lemma follows.
\end{proof}

%%%%%%%%%%%%%%%%%%%%%%%%%%%%%%%%%%%%%%%%%%%%%%%%%%%%%%%%
%%%%%%%%%%%%%%%%%%%%%%%%%%%%%%%%%%%%%%%%%%%%%%%%%%%%%%%%
%%%%%%%%%%%%%%%%%%%%%%%%%%%%%%%%%%%%%%%%%%%%%%%%%%%%%%%%

\bibliographystyle{plain}

\small{
\bibliography{bib}

\begin{thebibliography}{10}

\bibitem{berger1989ena}
B.~Berger, J.~Rompel, and P.~Shor.
\newblock Efficient {NC} algorithms for set cover with applications to learning
  and geometry.
\newblock {\em 30th Annual Symposium on the Foundations of Computer Science},
  pages 54--59, 1989.

\bibitem{lpsolve}
M.~Berkelaar, K.~Eikland, and P.~Notebaert.
\newblock {\em Lp\_solve mixed integer linear programming solver 5.5}, 2004.
\newblock Available at {\tt http://lpsolve.sourceforge.net/5.5}.

\bibitem{James01}
J.~Borneman, M.~Chrobak, G.D. Vedova, A.~Figueroa, and T.~Jiang.
\newblock Probe selection algorithms with applications in the analysis of
  microbial communities.
\newblock {\em Bioinformatics}, 17(1):S39--S48, 2001.

\bibitem{Feige98Threshold}
U.~Feige.
\newblock A threshold of ln n for approximating {Set Cover}.
\newblock {\em Journal of the ACM (JACM)}, 45(4):634--652, 1998.

\bibitem{Andres04}
A.~Figueroa, J.~Borneman, and T.~Jiang.
\newblock Clustering binary fingerprint vectors with missing values for {DNA}
  array data analysis.
\newblock {\em Journal of Computational Biology}, 11(5):887--901, 2004.

\bibitem{Garey79}
M.R. Garey and D.S. Johnson.
\newblock {\em Computers and Intractability. A Guide to the Theory of
  NP-Completeness}.
\newblock W.H.Freeman, New York, 1979.

\bibitem{gargano2003mgo}
L.~Gargano, AA. Rescigno, and U.~Vaccaro.
\newblock Multicasting to groups in optical networks and related combinatorial
  optimization problems.
\newblock {\em International Parallel and Distributed Processing Symposium},
  page 223, 2003.

\bibitem{Jampa05}
K.~Jampachaisri, L.~Valinsky, J.~Borneman, and S.~J. Press.
\newblock Classification of oligonucleotide fingerprints: application for
  microbial community and gene expression analyses.
\newblock {\em Bioinformatics}, 21(14):3122--3130, 2005.

\bibitem{Johan00}
J.~Schuchhardt, D.~Beule, A.~Malik, E.~Wolski, H.~Eickhoff, H.~Lehrach, and
  H.~Herzel.
\newblock Normalization strategies for {cDNA} microarrays.
\newblock {\em Nucleic Acids Res.}, 28(10):e47, 2000.

\bibitem{Valin04}
L.~Valinsky, A.~Scupham, G.D. Vedova, Z.~Liu, A.~Figueroa, K.~Jampachaisri,
  B.~Yin, E.~Bent, R.~Mancini-Jones, J.~Press, T.~Jiang, and J.~Borneman.
\newblock Oligonucleotide fingerprinting of ribosomal {RNA} genes {(OFRG)}.
\newblock In G.A. Kowalchuk, F.J. de~Bruijn, I.M. Head, A.D. Akkermans, and
  J.D.~Van Elsas, editors, {\em Molecular Microbial Ecology Manual}, pages
  569--585. Kluwer Academic Publishers, Dordrecht, The Netherlands, 2nd
  edition, 2004.

\bibitem{Valin02a}
L.~Valinsky, G.~Della Vedova, T.~Jiang, and J.~Borneman.
\newblock Oligonucleotide fingerprinting of ribosomal {RNA} genes for analysis
  of fungal community composition.
\newblock {\em Applied and Environmental Microbiology}, 68(12):5999--6004,
  2002.

\bibitem{Valin02b}
L.~Valinsky, G.~Della Vedova, A.~Scupham, S.~Alvey, A.~Figueroa, B.~Yin,
  R.~Hartin, M.~Chrobak, D.~Crowley, T.~Jiang, and J.~Borneman.
\newblock Analysis of bacterial community composition by oligonucleotide
  fingerprinting of {rRNA} genes.
\newblock {\em Applied and Environmental Microbiology}, 68(7):3243--3250, 2002.

\bibitem{young00kmedians}
N.~E. Young.
\newblock K-medians, facility location, and the {Chernoff-Wald} bound.
\newblock {\em ACM-SIAM Symposium on Discrete Algorithms}, pages 86--95, 2000.

\bibitem{Yu03}
W.~Yu, B.C. Ballif, C.D. Kashork, H.A. Heilstedt, L.A. Howard, W.W. Cai, L.D.
  White, W.~Liu, A.L. Beaudet, B.A. Bejjani, C.A. Shaw, and L.G. Shaffer.
\newblock Development of a comparative genomic hybridization microarray and
  demonstration of its utility with 25 well-characterized 1p36 deletions.
\newblock {\em Human Molecular Genetics}, 12(17):2145--2152, 2003.

\end{thebibliography}
}

%%%%%%%%%%%%%%%%%%%%%%%%%%%%%%%%%%%%%%%%%%%%%%%%%%%%%%%%
%%%%%%%%%%%%%%%%%%%%%%%%%%%%%%%%%%%%%%%%%%%%%%%%%%%%%%%%
%%%%%%%%%%%%%%%%%%%%%%%%%%%%%%%%%%%%%%%%%%%%%%%%%%%%%%%%

\end{document}